\documentclass{article}

\usepackage[preprint]{neurips_2021}
\usepackage[utf8]{inputenc}
\usepackage[T1]{fontenc}    
\usepackage{booktabs}  
\usepackage{nicefrac}
\usepackage{microtype} 
\usepackage{xcolor}  

\usepackage{hyperref}
\usepackage{url}

\usepackage{subcaption}
\usepackage{amsmath}
\usepackage{cleveref}
\usepackage{amssymb}
\usepackage{amsthm}
\usepackage{bbm}
\usepackage{amsfonts}
\usepackage{dsfont}
\usepackage{color}
\usepackage{tcolorbox}
\usepackage{arydshln}

\usepackage[ruled,vlined]{algorithm2e}
\usepackage{pifont}
\newcommand{\RN}[1]{%
	\textup{\uppercase\expandafter{\romannumeral#1}}%
}
\usepackage{graphicx}
\graphicspath{./home/marzsaee/Bureau/saeed/Thesis/Latex/}
\usepackage{subcaption}
\usepackage{listings}
\usepackage{color} 
\definecolor{mygreen}{RGB}{28,172,0} 
\definecolor{mylilas}{RGB}{170,55,241}
\usepackage{multirow}
\usepackage{tikz}

\usepackage{xr}
\makeatletter
\newcommand*{\addFileDependency}[1]{
  \typeout{(#1)}
  \@addtofilelist{#1}
  \IfFileExists{#1}{}{\typeout{No file #1.}}
}
\makeatother

\newcommand*{\myexternaldocument}[1]{
    \externaldocument{#1}
    \addFileDependency{#1.tex}
    \addFileDependency{#1.aux}
}
\myexternaldocument{Supplementary}

\newtheorem{proposition}{Proposition}[section]
\newtheorem{definition}{Definition}[section]

\newtheorem{remark}{Remark}[section]
\newtheorem{lemma}{Lemma}[section]

\newcommand{\mynewwidth}{.49}
\newcommand{\mynewscale}{.3}
\newcommand{\mybigwidth}{.6}
\newcommand{\mybigscale}{.33}

\newcommand{\algtitle}{Procedure}
\newcommand{\corrlayer}{\textit{Corr} layer }

\newif \iffinal
\finaltrue
\iffinal
\newcommand{\EDmodified}[1]{{#1}}
\newcommand{\EDcomments}[1]{{}}
\newcommand{\oldText}[1]{}
\newcommand{\SMmodified}[1]{{#1}}
\newcommand{\SMcomments}[1]{{}}
\else
\newcommand{\EDmodified}[1]{{\color{red} #1}}
\newcommand{\EDcomments}[1]{{\EDmodified{Erick commented: #1}}}
\newcommand{\oldText}[1]{}
\newcommand{\SMmodified}[1]{{\color{green} #1}}
\newcommand{\SMcomments}[1]{{\SMmodified{Saeed commented: #1}}}
\fi

\newcommand{\oldSectionThree}[1]{{}}

\newcommand{\quoteIt}[1]{``#1''}
\newcommand{\V}[1]{{\boldsymbol{#1}}}
\newcommand{\PW}{\V{w}}
\newcommand{\pw}{w}
\newcommand{\Areturns}{\V{\xi}}
\newcommand{\preturn}{\zeta}
\newcommand{\1}{\textbf{1}}
\newcommand{\removed}[1]{}
\newcommand{\tc}[2][black,fill=black]{\tikz[baseline=-0.5ex]\draw[#1,radius=#2] (0,0) circle ;}%
\newcommand{\circprod}{{\tc{1pt}}}
\newcommand{\circcomp}{{\circ}}
\newcommand{\myRe}{{\mathbb{R}}}

\title{\oldText{WaveCorr: A Dilated Convolution Based Model for Portfolio Management Using Reinforcement Learning}

WaveCorr: Correlation-savvy Deep Reinforcement Learning for Portfolio Management}
\author{
	Saeed Marzban \qquad \qquad Erick Delage\\
GERAD \& Department of Decision Sciences, HEC Montr\'{e}al, Montreal, Canada \\
	\texttt{saeed.marzban@hec.ca,erick.delage@hec.ca} 
\AND
Jonathan Yumeng Li\\
Telfer School of Management, University of Ottawa, Ottawa, Canada\\
\texttt{jonathan.li@telfer.uottawa.ca}
\AND
Jeremie Desgagne-Bouchard \qquad\qquad Carl Dussault\\
Evovest, Montreal, Canada
}
\date{April 2020}

\begin{document}

\maketitle

\begin{abstract}


The problem of portfolio management represents an important and challenging class of dynamic decision making problems, where rebalancing decisions need to be made over time with the consideration of many factors such as investors’ preferences, trading environments, and market conditions. In this paper, we present a new portfolio policy network architecture for deep reinforcement learning (DRL) that can exploit more effectively cross-asset dependency information and achieve better performance than  state-of-the-art architectures. In particular, we introduce a new property, referred to as \textit{asset permutation invariance}, for portfolio policy networks that exploit multi-asset time series data, and design the first portfolio policy network, named WaveCorr, that preserves this invariance property when treating asset correlation information. At the core of our design is an innovative permutation invariant correlation processing layer. An extensive set of experiments are conducted using data from both Canadian (TSX) and American stock markets (S\&P 500), and WaveCorr consistently outperforms other architectures with an impressive 3\%-25\% absolute improvement in terms of average annual return, and up to more than 200\% relative improvement in average Sharpe ratio. 
We also measured an improvement of a factor of up to 5 in the stability of performance under random choices of initial asset ordering and weights. The stability of the network has been found as particularly valuable by our industrial partner.

\end{abstract}


\section{Introduction}
In recent years, there has been a growing interest in applying Deep Reinforcement Learning (DRL) to solve dynamic decision problems that are complex in nature. One representative class of problems is portfolio management, whose formulation typically requires a large amount of continuous state/action variables and a sophisticated form of risk function for capturing the intrinsic complexity of financial markets, trading environments, and investors' preferences.

In this paper, we propose a new architecture of DRL for solving portfolio management problems that optimize a Sharpe ratio criterion. While there are several works in the literature that apply DRL for portfolio management problems such as \cite{moody1998performance, he2016deep,liang2018adversarial} among others, little has been done to investigate how to improve the design of a Neural Network (NN) in DRL so that it can capture more effectively the nature of dependency exhibited in financial data. In particular, it is known that extracting and exploiting cross-asset dependencies over time is crucial to the performance of portfolio management. The neural network architectures adopted in most existing works, such as Long-Short-Term-Memory (LSTM) or Convolutional Neutral Network (CNN), however, only process input data on an asset-by-asset basis and thus lack a mechanism to capture cross-asset dependency information. The architecture presented in this paper, named as WaveCorr, offers a mechanism to extract the information of both time-series dependency and cross-asset dependency. It is built upon the WaveNet structure  \citep{oord2016wavenet}, which uses dilated causal convolutions at its core, and a new design of correlation block that can process and extract cross-asset information.

In particular, throughout our development, we identify and define a property that can be used to guide the design of a network architecture that takes multi-asset data as input. This property, referred to as \textit{asset permutation invariance}, is motivated by the observation that the dependency across assets has a very different nature from the dependency across time. Namely, while the dependency across time is sensitive to the sequential relationship of data, the dependency across assets is not. To put it another way, given a multivariate time series data, the data would not be considered the same if the time indices are permuted, but the data should remain the same if the asset indices are permuted. While this property may appear more than reasonable, as discussed in section 3, a naive extension of CNN that accounts for both time and asset dependencies can easily fail to satisfy this property. To the best of our knowledge, the only other works that have also considered extracting cross-asset
dependency information in DRL for portfolio management are the recent works of \cite{zhang2020cost} and \cite{xu2020relation}.  While Zhang et al.’s work is closer to ours in that it is also built upon the idea of adding a correlation layer to a CNN-like module, its overall architecture is different from ours and, most noticeably, their design does not follow the property of asset permutation invariance and thus its performance can vary significantly when the ordering of assets changes. As further shown in the numerical section, our architecture, which has a simpler yet permutation invariant structure, outperforms in many aspects Zhang et al.’s architecture. The work of \cite{xu2020relation} takes a very different direction from ours, which follows a so-called attention mechanism and an encoder-decoder structure. A more detailed discussion is beyond the scope of this paper.

Overall, the contribution of this paper is three fold. First, we introduce a new property, referred to as asset permutation invariance, for portfolio policy networks that exploit multi-asset time series data. Second, we design the first portfolio policy network, named WaveCorr, that accounts for asset dependencies in a way that preserves this invariance. This achievement relies on the design of an innovative permutation invariant correlation processing layer. Third, and most importantly, we present evidence that WaveCorr significantly outperforms state-of-the-art policy network architectures using data from both Canadian (TSX) and American (S\&P 500) stock markets. Specifically, our new architecture leads to an impressive 5\%-25\% absolute improvement in terms of average annual return, up to more than 200\% relative improvement in average Sharpe ratio, and reduces, during the period of 2019-2020 (i.e. the Covid-19 pandemic), by 16\% the maximum daily portfolio loss  compared to the best competing method. Using the same set of hyper-parameters, we also measured an improvement of up to a factor of 5 in the stability of performance under random choices of initial asset ordering and weights, and observe that WaveCorr consistently outperforms our benchmarks under a number of variations of the model: including the number of available assets, the size of transaction costs, etc. Overall, we interpret this empirical evidence as a strong support regarding the potential impact of the WaveCorr architecture on automated portfolio management practices, and, more generally, regarding the claim that asset permutation invariance is an important NN property for this class of problems.

The rest of the paper unfolds as follows. Section \ref{sec:PM_problem} presents the portfolio management problem and risk averse reinforcement learning formulation. Section \ref{sec:architecture}  introduces the new property of \quoteIt{asset permutation invariance} for portfolio policy network and presents a new network architecture based on convolution networks that satisfies this property. Finally, Section \ref{sec:Experimental} presents the findings from our numerical experiments. We finally conclude in Section \ref{sec:Conclusion}.

\section{Problem statement}\label{sec:PM_problem}

\subsection{Portfolio management problem}

The portfolio management problem consists of optimizing the reallocation of wealth among many available financial assets including stocks, commodities, equities, currencies, etc. at discrete points in time. In this paper, we assume that there are $m$ risky assets in the market, hence the portfolio is controlled based on a set of weights $\PW_t \in \mathbb{W}:=\{\PW\in\mathbb{R}_+^{m}|\sum_{i=1}^m \pw^i=1\}$, which describes the proportion of wealth invested in each asset.
Portfolios are rebalanced at the beginning of each period $t=0,1,...,T-1$, which will incur proportional transaction costs for the investor, i.e. commission rates are of $c_s$ and $c_p$, respectively. We follow \cite{jiang2017deep} to model the evolution of the portfolio value and weights (see Figure \ref{fig:portfolio_evolution}). 
Specifically, during period $t$ the portfolio value and weights start at $p_{t-1}$ and $\PW_{t-1}$, and the changes in stock prices, captured by a random vector of asset returns $\Areturns_{t}\in \mathbb{R}^{m}$, affect the end of period portfolio value  $p^\prime_{t}:=p_{t-1} \Areturns_{t}^\top \PW_{t-1}$, and weight vector $\PW^\prime_{t}:=(p_{t-1}/p^\prime_t) \Areturns_{t}\circprod \PW_{t-1}$, where $\circprod$ is a term-wise product. The investor then decides on a new distribution of his wealth $\PW_{t}$, which triggers the following transaction cost:
\[c_s \sum_{i=1}^{m} (p_t' \pw'^{i}_{t} - p_t \pw^i_{t})^+ + c_p \sum_{i=1}^{m} (p_t \pw^i_{t} - p_t' \pw'^{i}_{t})^+\,.\]
Denoting the net effect of transaction costs on portfolio value with $\nu_t:=p_t/p_t'$, as reported in \cite{Li:TransactionCostOptim} one finds that $\nu_t$ is the solution of the following equations:
\[\nu_t = f(\nu_t,\PW_t',\PW_t):=1- c_s \sum_{i=1}^{m} (\pw'^{i}_{t} - \nu_t \pw^i_{t})^+ - c_p \sum_{i=1}^{m} (\nu_t \pw^i_{t} - \pw'^{i}_{t})^+.\]
This, in turn, allows us to express the portfolio's log return during the $t+1$-th period as:
\begin{equation}\label{eq:portfolioReturn}
    \preturn_{t+1} := \ln(p'_{t+1}/p'_{t}) = \ln(\nu_{t}p'_{t+1}/p_{t}) = \ln(\nu_t(\PW_{t}',\PW_{t}))+\ln(\Areturns_{t+1}^\top\PW_{t})
\end{equation}
where we make explicit the influence of $\PW_{t}'$ and $\PW_{t}$ on $\nu_t$.


\begin{figure}[h]
	\centering
	\includegraphics[scale=.33]{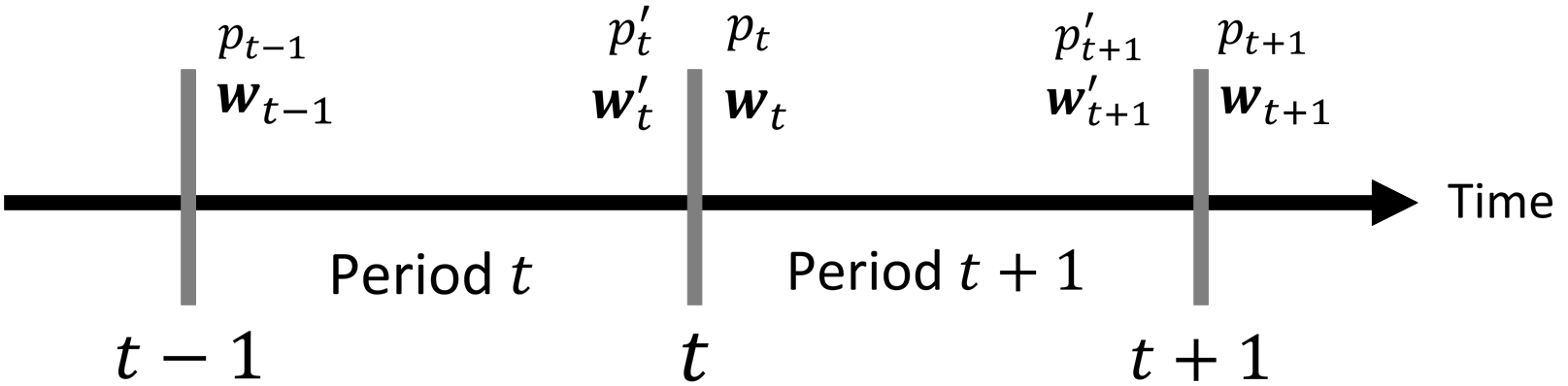}
	\caption{Portfolio evolution through time}
	\label{fig:portfolio_evolution}	
\end{figure}

We note that in \cite{jiang2017deep}, the authors suggest to approximate $\nu_t$ using an iterative procedure. However, we actually show in Appendix \ref{app:bissect} that $\nu_t$ can easily be identified with high precision using the bisection method.

\oldText{
\EDcomments{The old text is below. }
This will expose the investor to transaction costs that are inversely proportional to $\nu_t = \frac{p_t}{p'_{t}} \in (0,1]$. The higher the value of $\nu_t$, the lower the transaction costs paid by the investor. $\nu_t$ is a function of $\PW_t, \PW'_{t-1}$ and the purchase and sell trade fees. Given that $\nu_t$ is known to satisfy:
\begin{equation}
    \nu_t = f(\nu_t):=\frac{1}{1-c_p \pw^0_{t}} \left[ 1- c_p \pw'^{0}_{t} - (c_s + c_p - c_s c_p) \sum_{i=1}^{m} (\pw'^{i}_{t} - \nu_t \pw^i_{t})^+ \right]\,,
\end{equation}

\[c_s \sum_{i=1}^{m} (\pw'^{i}_{t} - \nu_t \pw^i_{t})^+ + c_p' \sum_{i=1}^{m} (\nu_t \pw^i_{t} - \pw'^{i}_{t})^++\nu_t = 1\]

\[c_s \sum_{i=1}^{m} (\pw'^{i}_{t} - p_t \pw^i_{t})^+ + c_p' \sum_{i=1}^{m} (p_t \pw^i_{t} - \pw'^{i}_{t})^++p_t = p_t'\]

\[c_p'=c_p/(1-c_p)\]

 \cite{jiang2017deep} suggests to approximate $\nu_t$ using a fixed number of iterations $\nu_t\approx \nu_t^k$, where $\nu_t^k =f(\nu_t^{k-1})$ and $\nu_t^0$ is fixed. Convergence for large enough $k$ is guaranteed by virtues of the monotone convergence theorem. They suggest calibrating the number of needed iterations using a validation set. In Appendix \ref{thm:bissectThm}, we actually show that $\nu_t$ can be identified with extremely high precision using the bisection method. This in turn allows us to streamline the overall training procedure by removing one hyper-parameter to tune.

\begin{figure}[h]
	\centering
	\includegraphics[scale=.9]{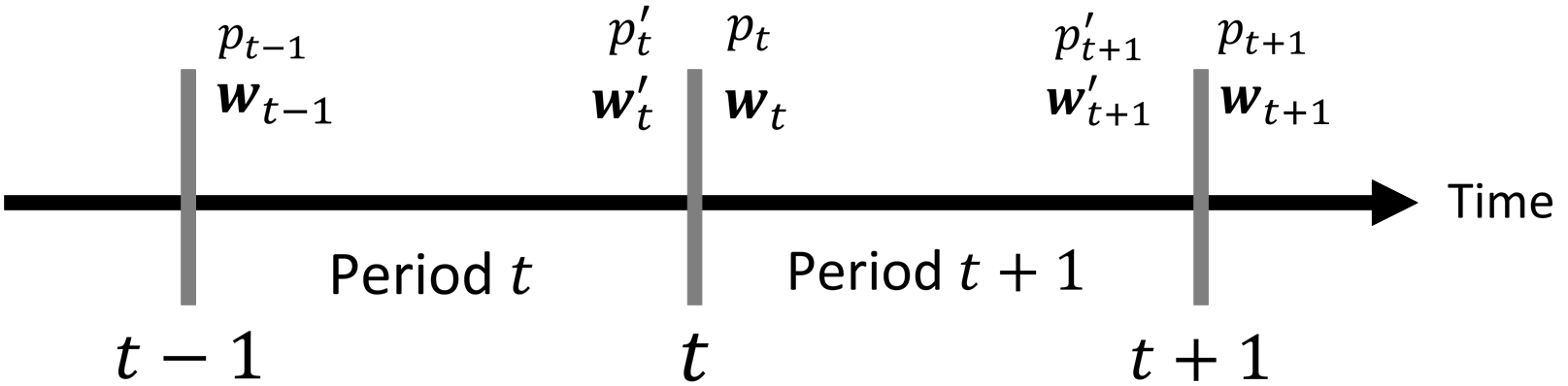}
	\caption{Portfolio evolution through time}
	\label{fig:portfolio_evolution}	
\end{figure}

Having $\nu_t$, we can compute $\bar{r}^p_{t+1}$ that is the portfolio log return during period $[t,t+1)$ as follows:
\begin{equation}\label{eq:portfolioReturn}
    \bar{r}^p_{t+1} :=  \ln{\frac{p'_{t}}{p'_{t-1}}} = \ln{\frac{\nu_{t}p'_{t}}{p_{t}}} = \ln{[\nu_t(\PW_{t}) \Areturns_{t+1} \cdot \PW_{t}]}
\end{equation}
where we made explicit the influence of $\PW_{t}$ on $\nu_t$, and $\Areturns_{t} \in \mathbb{R}^{m}$ is the vector of asset returns during time period $[t,t+1)$.
}

\subsection{Risk-averse Reinforcement Learning Formulation}\label{sec:RL}

In this section, we formulate the portfolio management problem as a Markov Decision Process (MDP) denoted by $(\mathcal{S},\mathcal{A},r,P)$. In this regard, the agent (i.e. an investor) interacts with a stochastic environment by taking an action $a_t \equiv \PW_t\in \mathbb{W}$ after observing the state $s_t \in \mathcal{S}$ composed of a window of historical market observations, which include the latest stock returns $\Areturns_{t}$, along with the final portfolio composition of the previous period $\PW_t'$. This action results in the immediate stochastic reward that takes the shape of an approximation of the realized log return, i.e. $r_{t}(s_t,a_t,s_{t+1}) := \ln(f(1,\PW_t',\PW_t))+\ln(\Areturns_{t+1}^\top \PW_t)\approx\ln(\nu(\PW_t',\PW_t))+\ln(\Areturns_{t+1}^\top \PW_t) $, for which a derivative is easily obtained.  
Finally, $P$ captures the assumed Markovian transition dynamics of the stock market and its effect on portfolio weights: $P(s_{t+1}|s_0,a_0,s_1,a_1,...,s_t,a_t)=P(s_{t+1}|s_t,a_t)$. 



Following the works of \cite{moody1998performance} and \cite{almahdi2017adaptive} on risk averse DRL, our objective is to identify a deterministic trading policy $\mu_\theta$ (parameterized by $\theta$) that maximizes the expected value of the Sharpe ratio measured on $T$-periods log return trajectories generated by $\mu_\theta$. Namely:
\begin{eqnarray} \label{mainprob}
\max_
\theta\; J_F(\mu_\theta):=\mathbb{E}_{\substack{s_0\sim F\\s_{t+1}\sim P(\cdot|s_t,\mu_\theta(s_t))}}[SR(r_0(s_0,\mu_\theta(s_0),s_1),...,r_{T-1}(s_{T-1},\mu_\theta(s_{T-1}),s_T))]
\end{eqnarray}
where $F$ is some fixed distribution and
\begin{equation*}
\begin{aligned}
SR(r_{0:T-1}) &:= \frac{(1/T)\sum_{t=0}^{T-1} r_t}{\sqrt{(1/(T-1))\sum_{t=0}^{T-1} (r_t-(1/T)\sum_{t=0}^{T-1} r_t)^2}}\,.
\end{aligned}
\end{equation*}
The choice of using the Sharpe ratio of log returns is motivated by modern portfolio theory (see \cite{markowitz1952portfolio}), 
which advocates a balance between expected returns and exposure to risks, and where it plays the role of a canonical way of exercising this trade-off \citep{sharpeRatio}. While it is inapt of characterizing downside risk, it is still considered a ``gold standard of performance evaluation" by the financial community \citep{bailey2012sharpe}. In \cite{moody1998performance}, the trajectory-wise Sharpe ratio is used as an estimator of the instantaneous one in order to facilitate its use in RL. A side-benefit of this estimator is to offer some control on the variations in the evolution of the portfolio value which can be reassuring for the investor.

\oldText{
\SMmodified{This can provide the model with a suitable proxy that can accordingly optimize the Sharpe ratio over the test data. Apart from that, using Sharpe ratio in this study is beneficial in that there will be one less hyper-parameter to worry about when we compare the performance of our proposed model with other RL-based models. In particular, using performance measures in which the risk aversion can be adjusted using some parameters could create circumstances under which some of the models outperform the other ones only for specific values of risk aversion. This makes the comparison of these models a difficult task. Using Sharpe ratio that is widely accepted among practitioners will resolve this issue.}

\SMmodified{Sharpe ratio is one of the most widely used performance evaluation measures by practitioners that is in line with the modern portfolio theory where risk-adjusted returns are extensively studied. Contrary to some deficiencies that are assigned to this measure such as inability in capturing the downside risk, its simplicity has convinced the financial community to widely use it so that it is called the ``gold standard of performance evaluation" \citep{bailey2012sharpe}. it measures the average portfolio return in a specified time window in excess of the risk free return per unit of volatility, where by subtracting the risk free return, the measure is able of well-demonstrating the profits that are corresponding to risky activities as apposed to risk free ones. Sharpe ratio is assessing if the extra return obtained from an investment policy is due to smart financial decisions or it comes at the cost of having higher risk (volatility).} }


In the context of our portfolio management problem, since $s_t$ is composed of an exogeneous component $s_t^{exo}$ which includes $\Areturns_{t}$ and an endogenous state $\PW_t'$ that becomes deterministic when $a_t$ and $s_{t+1}^{exo}$ are known, we have that:
\[J_F(\mu_\theta):=\mathbb{E}_{\substack{s_0\sim F\\s_{t+1}\sim P(\cdot|s_t,\beta(s_t)))}}[SR(r_0(\bar{s}_0^\theta,\mu_\theta(\bar{s}_0^\theta),\bar{s}_1^\theta)
,\dots,r_{T-1}(\bar{s}_{T-1}^\theta,\mu_\theta(\bar{s}_{T-1}^\theta),\bar{s}_T^\theta))]\]
where $\beta(s_t)$ is an arbitrary policy, and where the effect of $
\mu_\theta$ on the trajectory can be calculated using 
\[\bar{s}_t^\theta:=\left(s_{t}^{exo},\frac{\Areturns_{t}\circprod\mu_\theta(\bar{s}_{t-1}^\theta)}{\Areturns_{t}^\top \mu_\theta(\bar{s}_{t-1}^\theta)}\right)\,,\]
for $t\geq 1$, 
while $\bar{s}_0^\theta:=s_0$. 
Hence,
\begin{eqnarray}\label{eq:SGD}
\nabla_\theta J_F(\mu_\theta):=\mathbb{E}
[\nabla_\theta SR(r_0(\bar{s}_0^\theta,\mu_\theta(\bar{s}_0^\theta),\bar{s}_1^\theta),
\dots,r_{T-1}(\bar{s}_{T-1}^\theta,\mu_\theta(\bar{s}_{T-1}^\theta),\bar{s}_T^\theta))]\,,
\end{eqnarray}
where $\nabla_\theta SR$ can be obtained by backpropagation using the chain rule. This leads to the following stochastic gradient step:
\[    \theta_{k+1} = \theta_k + \alpha \nabla_\theta SR(r_0(\bar{s}_0^\theta,\mu_\theta(\bar{s}_0^\theta),\bar{s}_1^\theta), \dots,r_{T-1}(\bar{s}_{T-1}^\theta,\mu_\theta(\bar{s}_{T-1}^\theta),\bar{s}_T^\theta)\,,\]
with $\alpha>0$ as the step size.

\oldText{
\EDcomments{OLD TEXT BELOW}

where we also suppress the inclusion of risk free returns in the definition of this Sharpe ratio, and consider an unbiased estimation of the standard deviation. Now, if we consider the policy to be a function of the parameter set $\theta$ in a neural network, $a_t=\mu_\theta(s_t)$, the goal would be to find a sub-optimal $\theta$ with regard to the Sharpe ratio objective function. The whole process of performing the portfolio management is shown in Figure \ref{fig:AOmodel}, where the agent chooses a batch of size $T$, through this batch he sequentially observes the states, makes an action by using the parameterized policy $\mu_\theta(s_t)$, then uses this action together with the past price series from the environment to form the state $s_{t+1}$ and decides on the action in the next period. This will provide him with a sequence of rewards by which he can compute the $SR$ function in a forward pass, and then update the parameter set $\theta$ in a backward pass by using gradient ascent:
\begin{equation}
    \theta = \theta + \alpha \nabla_\theta SR(r_1(s_0,\mu(s_0)),...,r_T(s_{T-1},\mu(s_{T-1})))
\end{equation}
}

\oldSectionThree{
\section{WaveCorr}

There are several considerations that go into the design of the network for the trading policy function $\mu_{\theta}$. First, the network should have the capacity to handle long time series data, which allows for extracting long-term dependencies. Second, the network should be flexible in its design for capturing dependencies across time and assets. Third, the network should not be overly-parameterized when applied to longer time series and increasing number of assets. A natural choice of the architecture for time series analysis is CNN, which is known to outperform recurrent architectures in sequence modeling \citep{bai2018empirical}. The mechanism of kernel as the building block of causal convolutions provides remarkable flexibility in devising new components for better capturing the characteristics of input time series. However, for longer time series, causal convolutions could suffer from the issue of over-parameterization ( over-fitting) or gradient explosion/vanishing due to the fast-growing size/number of kernels in use. In this study, we adopt the WaveNet structure \citep{oord2016wavenet}, which is to use \textit{dilated} causal convolutions along with residual and skip connections to increase the receptive field of the network. Inspired by similar dilated convolution based models \citep{bai2018empirical,kalchbrenner2016neural,lea2017temporal} we simplify the structure by removing the gated activation units, conditioning, and context stacks.


Directly applying the WaveNet model to our portfolio management problem will solely capture dependency across time, but overlook the cross-asset correlation information. This is because the convolutions embedded in the WaveNet mdoel are 1D. There are several reasons why extending a convolutional architecture like WaveNet to capture cross-asset dependencies is not straightforward and requires great care. First, extending to 2D convolutions would significantly increase the number of parameters in the model, which makes it more prone to the issue of over-fitting. This is specifically an important problem when the network is used in an RL framework that is already notorious for having high potential of extreme over-fitting. Second, and most importantly, simply extending to 2D convolutions would lead to a network that is sensitive to the ordering of assets in input data.

Our goal is to design a network that can, on one hand, leverage on the superior performance of WaveNet in capturing time dependency, and on the other hand, can effectively extract cross-asset dependency without becoming overly parameterized. We first present the general architecture of WaveCorr in Figure \ref{fig:generalArc}.

\begin{figure}[h]
	\centering
	\includegraphics[scale=.5]{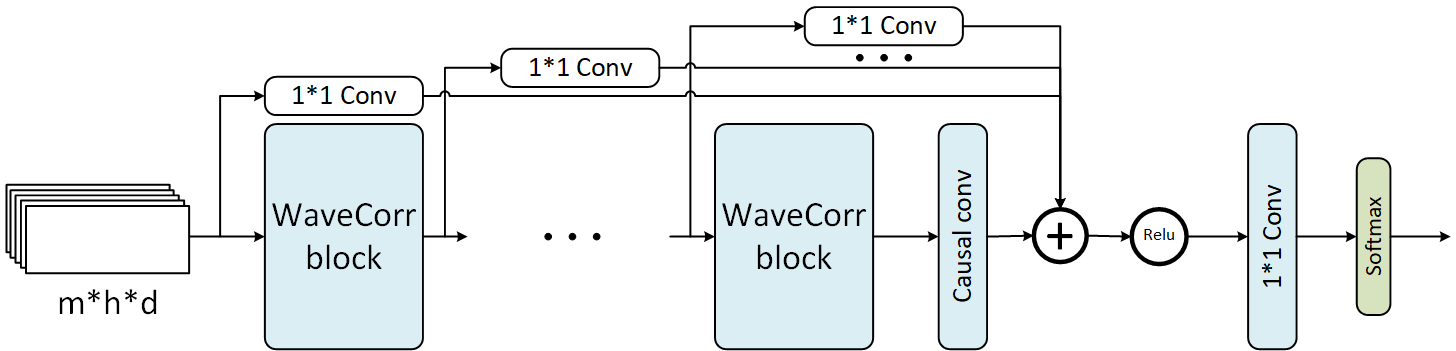}
	\caption{The architecture of the WaveCorr model}
	\label{fig:generalArc}	
\end{figure}

Here, the network takes as input a tensor of dimension $m\times h \times d$, where $m:$ the number of assets, $h:$ the size of look-back time window, $d:$ the number of channels (number of features for each asset), and generates as output an $m$-dimensional wealth allocation vector. The WaveCorr blocks, which play the key role for extracting cross time/asset dependencies, form the body of the architecture. In order to provide more flexibility for the choice of $h$, we define a causal convolution after the sequence of WaveCorr blocks to adjust the receptive field so that it includes the whole length of the input time series. Also, similar to the WaveNet structure, we use skip connections in our architecture. The structure of the WaveCorr (residual) block is presented in Figure \ref{fig:WaveCorr}.

\begin{figure}[h]
	\centering
	\includegraphics[scale=.5]{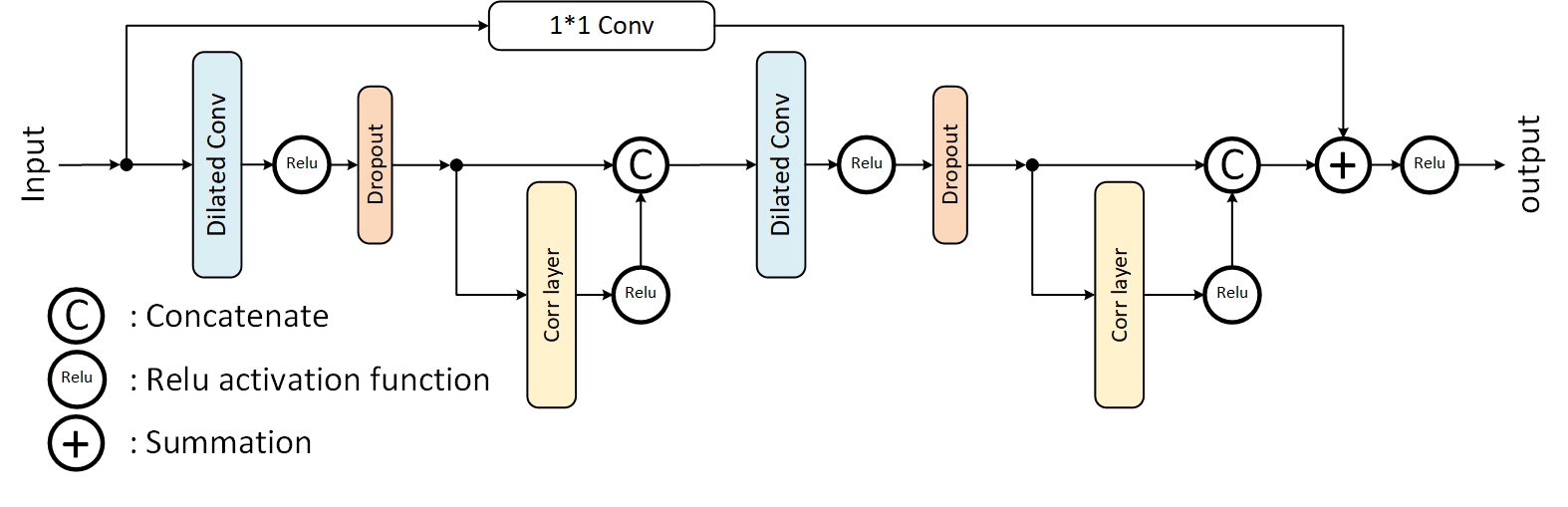}
	\caption{WaveCorr residual block}
	\label{fig:WaveCorr}	
\end{figure}

The design of the residual block in WaveCorr extends a simplified variation of the residual block in WaveNet by adding the correlation layers (and Relu, concatenation operations follow right after). In particular, the block includes two layers of dilated causal convolutions followed by Relu activation functions and dropout layers. Having an input of dimensions $m\times h \times d$, the convolutions output tensors of dimension $m \times h \times d'$ where each slice of the output tensor, i.e. an $m \times 1 \times d$ matrix, contains the dependency information of a variable (an asset) over time. By applying different dilation rates in each WaveCorr block, the model is able of extracting the dependency information for a longer time horizon. A dropout layer with a rate of $50\%$ is considered to prevent over-fitting, whereas for the gradient explosion/vanishing prevention mechanism of residual connection we use a $1 \times 1$ convolution on the top of Figure \ref{fig:WaveCorr} to ensure the summation operation is over tensors of the same shape. The \corrlayer generates an output tensor of dimensions $m\times h \times 1$ from an $m\times h \times d$ input, but each slice of the output tensor, i.e. an $m \times 1 \times 1$ matrix, should contain cross-asset dependency information.
The concatenation operator combines the cross-asset dependency information obtained from \textit{Corr} layer and the cross-time dependency information obtained before applying the convolutions. The output is then fed into another dilated convolution for extracting dependency information for a longer time horizon.

We propose here a property that can be used to further guide the design, namely 
the property of asset permutation invariance. Let $\sigma: \mathbb{R}^{m\times h \times d} \rightarrow \mathbb{R}^{m \times h \times d}$ denote a permutation operator over a tensor $\mathcal{T}$ such that $\sigma(\mathcal{T})[i,:,:] = \mathcal{T}[\pi(i),:,:]$, where $\pi: \{1,...,m\} \rightarrow \{1,...,m\}$ is a bijective function. The operator is overloaded for a set of tensors ${\cal S} \subseteq \mathbb{R}^{m \times h \times d}$ such that $\sigma({\cal S}):=\{\sigma(\mathcal{T})\}_{\mathcal{T} \in {\cal S}} \subseteq \mathbb{R}^{m\times h \times d}$. 

\begin{definition} (Asset Permutation Invariance) \label{pi}
A block, which takes an input tensor $\mathcal{T} \in \mathbb{R}^{m \times h \times d}$ and generates an output tensor $\mathcal{T}^o \in \mathbb{R}^{m \times h' \times d'}$, is permutation invariant if given any two input tensors $\mathcal{T}, \mathcal{T}' \in  \mathbb{R}^{m \times h \times d}$ that satisfy $\mathcal{T}=\sigma(\mathcal{T}')$, the corresponding set of output tensors that can possibly be generated from the block, denoted respectively by ${\cal S}$, ${\cal S}'$ satisfy $\mathcal{S} = \sigma({\mathcal{S}')}$.
\end{definition}

This definition is motivated by the idea that the final outputs should differ only in the ordering, but not the values, of their entries when a different ordering of assets is applied in the input tensor. One can verify, for instances, that a dilated convolution block is permutation invariant or a block consisting of a number of permutation invariant sub-blocks is permutation invariant (see Appendix \ref{sec:app:proofSecArch}). We detail now our design of the \corrlayer via \algtitle \ref{alg:corr_layer}. The operator $CC$ maps $\mathbb{R}^{(m+1) \times h \times d} \rightarrow \mathbb{R}^{1 \times h \times 1}$ by using a kernel of size $(m+1) \times 1$, and the operator $Concat_i$ concatenates two tensors on the $i$-th dimension


\begin{algorithm}
\caption{\textit{Corr} layer}
\label{alg:corr_layer}
\SetAlgoLined
\KwResult{Tensor that contains correlation information, $\mathcal{O}_{out}$ of dimension $m \times h \times 1$}
\textbf{Inputs}: Tensor $\mathcal{O}_{in}$ of dimension $m \times h \times d$\;
Define an empty tensor $\mathcal{O}_{out} = Null$\;
\For{$i=1:m$}
{
$\mathcal{O}_{mdl}=Concat_1(\mathcal{O}_{in}[i,:,:],\mathcal{O}_{in}) \in \mathbb{R}^{(m+1) \times h \times d}$\;
\eIf{$\mathcal{O}_{out} = Null$}{$ \mathcal{O}_{out} = CC(\mathcal{O}_{mdl})$}
{
Set $\mathcal{O}_{out} =Concat_1(\mathcal{O}_{out},CC(\mathcal{O}_{mdl}))$
}
}
\end{algorithm}

In the above \algtitle, the kernel is applied to a tensor ${\cal O}_{mdl}$ constructed from adding $i$-th row of the input tensor on the top of the input tensor. Concatenating the output tensors from each run gives the final output tensor. The above \algtitle ensures that the output has the same height $m$ as the input. Figure \ref{fig:corrLayer} gives an example for a $5\times 1\times 1$ input tensor (5 assets). We show in Appendix C that based on the \corrlayer, the WaveCorr residual block (Figure \ref{fig:WaveCorr}) satisfies the property of asset permutation invariance.

\begin{figure}[h]
	\centering
	\includegraphics[scale=.5]{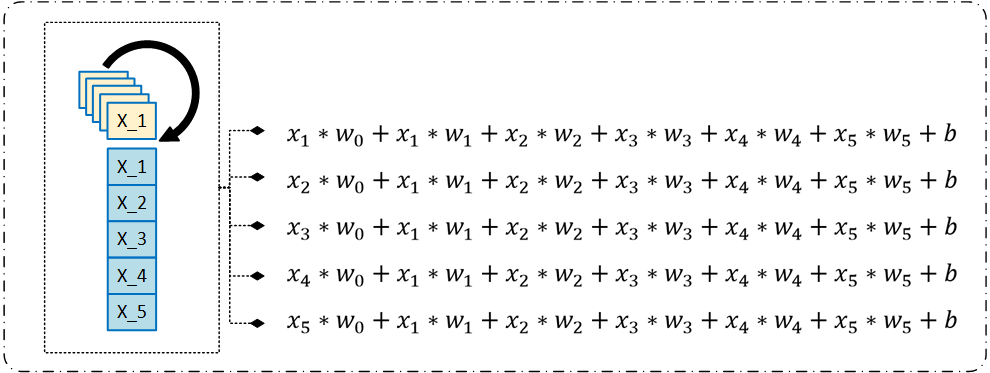}
	\caption{An example of the \textit{Corr} layer over 5 assets}
	\label{fig:corrLayer}
\end{figure}


For our portfolio management problem, we consider a look-back window of 30 days as input to the model, which leads us to consider three WaveCorr blocks in Figure \ref{fig:generalArc}. In addition, for efficiency, we consider one, instead of two, \corrlayer per residual block. The details of each layer in terms of the kernel size, number of channels, dilation rates, and types of activation functions are shown in Table \ref{tbl:Wavenet_structure}.

\begin{table}[h]
\caption{The structure of the network}
\centering
\begin{tabular}{lccccc}
\toprule
Layer & Input shape & Output shape & Kernel & Activation & Dilation rate \\
\midrule
Dilated conv & $(m \times h \times d)$ & $(m \times h \times 8)$ & $(1 \times 3)$ & Relu & 1\\
Dilated conv & $(m \times h \times 8)$ & $(m \times h \times 8)$ & $(1 \times 3)$ & Relu & 1\\
\textit{Corr} layer & $(m \times h \times 8)$ & $(m \times h \times 1)$ & $([m+1] \times 1)$ & Relu & -\\

\cmidrule(r){1-6}

Dilated conv & $(m \times h \times 9)$ & $(m \times h \times 16)$ & $(1 \times 3)$ & Relu & 2\\
Dilated conv & $(m \times h \times 16)$ & $(m \times h \times 16)$ & $(1 \times 3)$ & Relu & 2\\
\textit{Corr} layer & $(m \times h \times 16)$ & $(m \times h \times 1)$ & $([m+1] \times 1)$ & Relu & -\\

\cmidrule(r){1-6}

Dilated conv & $(m \times h \times 17)$ & $(m \times h \times 16)$ & $(1 \times 3)$ & Relu & 4\\
Dilated conv & $(m \times h \times 16)$ & $(m \times h \times 16)$ & $(1 \times 3)$ & Relu & 4\\
\textit{Corr} layer & $(m \times h \times 16)$ & $(m \times h \times 1)$ & $([m+1] \times 1)$ & Relu & -\\

\cmidrule(r){1-6}

Causal conv & $(m \times h \times 17)$ & $(m \times h \times 16)$ & $(1 \times [h-28])$ & Relu & -\\
$1 \times 1$ conv & $(m \times h \times 16)$ & $(m \times h \times 1)$ & $(1 \times 1)$ & Softmax & -\\
\bottomrule

\end{tabular}
\label{tbl:Wavenet_structure}
\end{table}

We should mention here a recent work of \cite{zhang2020cost}, where the authors propose an architecture that also take both sequential and cross-asset dependency into consideration. Their architecture, from a high level perspective, is more complex than ours (Figure \ref{fig:generalArc}) in that theirs involves two networks, one LSTM and one CNN, whereas ours is built solely on CNN. Our architecture is thus simpler to implement, easier to backtest, and allows for more efficient computation. The most noticeable difference between their design and ours is at the level of \corrlayer. They use an $m \times 1$ kernel to extract dependency across assets and apply a standard padding trick to keep the input tensor invariant in size. Their approach suffers from two issues: first, the kernel in their design may capture only partial dependency information (see the appendix for a more detailed discussion), and second, most problematically, their design is not permutation invariant and thus the performance of their network can be highly sensitive to the ordering of assets. This second issue is further demonstrated in the next section.

}

\section{The New Permutation Invariant WaveCorr Architecture}\label{sec:architecture}

There are several considerations that go into the design of the network for the portfolio policy network $\mu_{\theta}$.  First, the network should have the capacity to handle long historical time series data, which allows for extracting long-term dependencies across time. Second, the network should be flexible in its design for capturing dependencies across a large number of available assets. Third, the network should be parsimoniously parameterized to achieve these objectives without being prone to overfitting. To this end, the WaveNet structure \citep{oord2016wavenet} offers a good basis for developing our architecture and was employed in \cite{zhang2020cost}. Unfortunately, a direct application of WaveNet in portfolio management struggles at processing the cross-asset correlation information. This is because the convolutions embedded in the WaveNet model are 1D and extending to 2D convolutions increases the number of parameters in the model, which makes it more prone to the issue of over-fitting, a notorious issue particuarly in RL. Most importantly, naive attempts at adapting WaveNet to account for such dependencies (as done in \cite{zhang2020cost}) 
can make the network become sensitive to the ordering of the assets in the input data, an issue that we will revisit below.




We first present the general architecture of WaveCorr in Figure \ref{fig:generalArc}.
\begin{figure}[h]
	\centering
	\includegraphics[scale=.5]{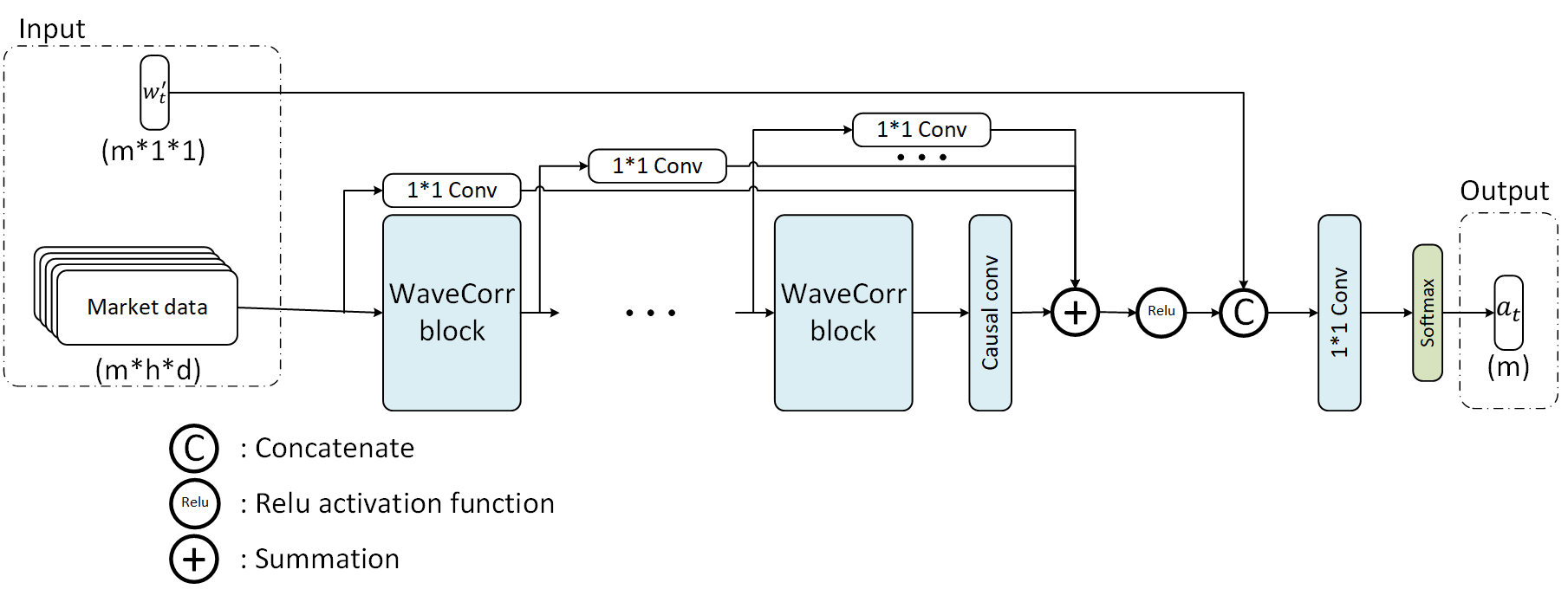}
	\caption{The architecture of the WaveCorr policy network}
	\label{fig:generalArc}	
\end{figure}
Here, the network takes as input a tensor of dimension $m\times h \times d$, where $m:$ the number of assets, $h:$ the size of look-back time window, $d:$ the number of channels (number of features for each asset), and generates as output an $m$-dimensional wealth allocation vector. The WaveCorr blocks, which play the key role for extracting cross time/asset dependencies, form the body of the architecture. In order to provide more flexibility for the choice of $h$, we define a causal convolution after the sequence of WaveCorr blocks to adjust the receptive field so that it includes the whole length of the input time series. Also, similar to the WaveNet structure, we use skip connections in our architecture. 

\begin{figure}[h]
	\centering
	\includegraphics[scale=.5]{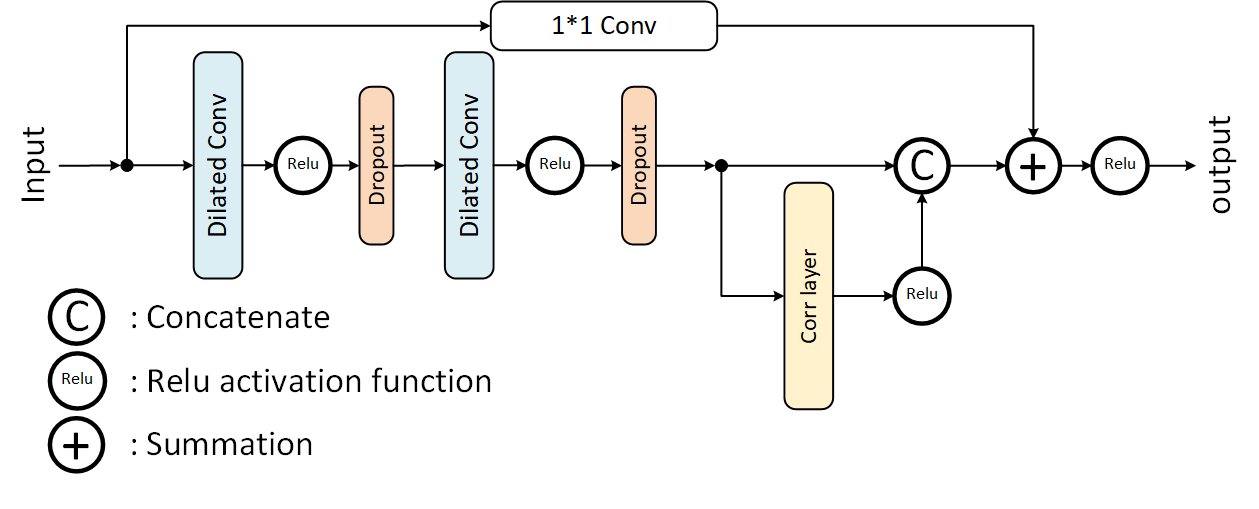}
	\caption{WaveCorr residual block}
	\label{fig:WaveCorr}	
\end{figure}

\EDcomments{Could this network could model for asset $i$:
\[E[((\xi_i-\mu_i) - (\xi_j-\mu_j))^+ ]\]
or
\[E[((\xi_j-\mu_j)- (\xi_i-\mu_i))^+ ]\]
}
The design of the WaveCorr residual block in WaveCorr extends a simplified variation \citep{bai2018empirical} of the residual block in WaveNet  
by adding our new correlation layers (and Relu, concatenation operations following right after). As shown in Figure \ref{fig:WaveCorr}, the block includes two layers of dilated causal convolutions followed by Relu activation functions and dropout layers. Having an input of dimensions $m\times h \times d$, the convolutions output tensors of dimension $m \times h \times d'$ where each slice of the output tensor, i.e. an $m \times 1 \times d$ matrix, contains the dependency information of each asset over time. By applying different dilation rates in each WaveCorr block, the model is able of extracting the dependency information for a longer time horizon. A dropout layer with a rate of $50\%$ is considered to prevent over-fitting, whereas for the gradient explosion/vanishing prevention mechanism of residual connection we use a $1 \times 1$ convolution (presented on the top of Figure \ref{fig:WaveCorr}), which inherently ensures that the summation operation is over tensors of the same shape. The \corrlayer generates an output tensor of dimensions $m\times h \times 1$ from an $m\times h \times d$ input, where each slice of the output tensor, i.e. an $m \times 1 \times 1$ matrix, is meant to contain cross-asset dependency information.
The concatenation operator combines the cross-asset dependency information obtained from the \corrlayer with the cross-time dependency information obtained from the causal convolutions. 

Before defining the \corrlayer, we take a pause to introduce a property that will be used to further guide its design, namely 
the property of asset permutation invariance. This property is motivated by the idea that the set of possible investment policies that can be modeled by the portfolio policy network should not be affected by the way the assets are indexed in the problem. On a block per block level, we will therefore impose that, when the asset indexing of the input tensor is reordered, the set of possible mappings obtained should also only differ in its asset indexing. More specifically, we let $\sigma: \mathbb{R}^{m\times h \times d} \rightarrow \mathbb{R}^{m \times h \times d}$ denote a permutation operator over a tensor $\mathcal{T}$ such that $\sigma(\mathcal{T})[i,:,:] = \mathcal{T}[\pi(i),:,:]$, where $\pi: \{1,...,m\} \rightarrow \{1,...,m\}$ is a bijective function. Furthermore, we consider $\sigma^{-1}: \mathbb{R}^{m\times h \times d'} \rightarrow \mathbb{R}^{m \times h \times d'}$ denote its \quoteIt{inverse} such that $\sigma^{-1}(\mathcal{O})[i,:,:]:=\mathcal{O})[\pi^{-1}(i),:,:]$, with $\mathcal{O}\in\mathbb{R}^{m\times h \times d'}$.

\begin{definition} (Asset Permutation Invariance) \label{pi}
A block capturing a set of functions $\mathcal{B}\subseteq \{B:\mathbb{R}^{m \times h \times d}\rightarrow \mathbb{R}^{m \times h' \times d'}\}$ is \textbf{asset permutation invariant} if given any permutation operator $\sigma$, we have that $\{\sigma^{-1}\circcomp B\circcomp \sigma: B\in\mathcal{B}\} = \mathcal{B}$, where $\circcomp$ stands for function composition.
\end{definition}

One can verify, for instances, that all the blocks described so far in WaveCorr are permutation invariant and that asset permutation invariance is preserved under composition (see Appendix \ref{sec:app:proofPI}). 

With this property in mind, we can now detail the design of a permutation invariant \corrlayer via \algtitle \ref{alg:corrlayer}, where we denote as $CC:\mathbb{R}^{(m+1) \times h \times d} \rightarrow \mathbb{R}^{1 \times h \times 1}$ the operator that applies an $(m+1)\times 1$ convolution, and as $Concat_1$  the operator that concatenates two tensors along the first dimension.
In Procedure \ref{alg:corrlayer}, the kernel is applied to a tensor ${\cal O}_{mdl}\in \mathbb{R}^{(m+1) \times h \times d}$ constructed from adding the $i$-th row of the input tensor on the top of the input tensor. Concatenating the output tensors from each run gives the final output tensor.
Figure \ref{fig:corrLayer} gives an example for the case with $m=5$, and $h=d=1$. Effectively, one can show that \corrlayer satisfies asset permutation invariance (proof in Appendix).

\begin{algorithm}[h]
\caption{\textit{Corr} layer\label{alg:corrlayer}}
\label{alg:corr_layer}
\SetAlgoLined
\KwResult{Tensor that contains correlation information, $\mathcal{O}_{out}$ of dimension $m \times h \times 1$}
\textbf{Inputs}: Tensor $\mathcal{O}_{in}$ of dimension $m \times h \times d$\;
Define an empty tensor $\mathcal{O}_{out}$ of dimension $0\times h\times 1$\;
\For{$i=1:m$}
{
Set $\mathcal{O}_{mdl}=Concat_1(\mathcal{O}_{in}[i,:,:],\mathcal{O}_{in}) $\;
Set $\mathcal{O}_{out} =Concat_1(\mathcal{O}_{out},CC(\mathcal{O}_{mdl})))$\;
}
\end{algorithm}

\begin{proposition}\label{thm:corrlayer}
The \corrlayer block satisfies asset permutation invariance.
\end{proposition}

Table \ref{tbl:Wavenet_structure} summarizes the details of each layer involved in the WaveCorr architecture: including kernel sizes, internal numbers of channels, dilation rates, and types of activation functions. Overall, the following proposition confirms that this WaveCorr portfolio policy network satisfies asset permutation invariance (see Appendix for proof). 

\begin{proposition}\label{thm:permInv}
The WaveCorr portfolio policy network architecture satisfies asset permutation invariance.
\end{proposition}

\begin{figure}[h]
	\centering
	\includegraphics[scale=.5]{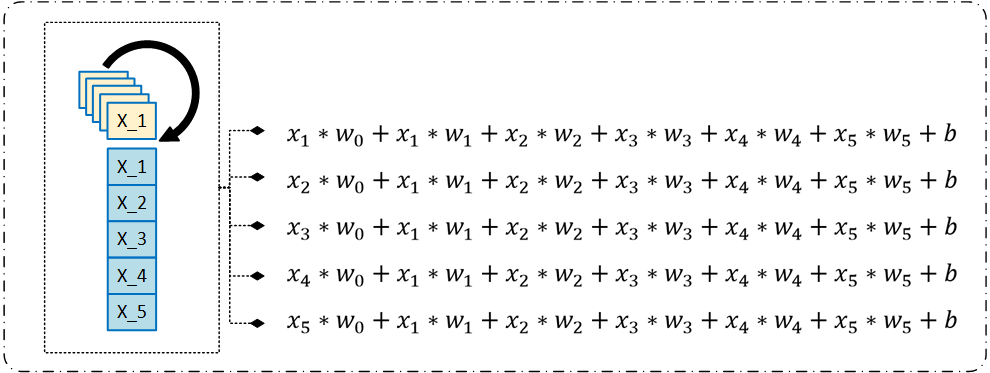}
	\caption{An example of the \textit{Corr} layer over 5 assets}
	\label{fig:corrLayer}
\end{figure}

\begin{table}[h]
\caption{The structure of the network}
\centering
\begin{tabular}{lccccc}
\toprule
Layer & Input shape & Output shape & Kernel & Activation & Dilation rate \\
\midrule
Dilated conv & $(m \times h \times d)$ & $(m \times h \times 8)$ & $(1 \times 3)$ & Relu & 1\\
Dilated conv & $(m \times h \times 8)$ & $(m \times h \times 8)$ & $(1 \times 3)$ & Relu & 1\\
\textit{Corr} layer & $(m \times h \times 8)$ & $(m \times h \times 1)$ & $([m+1] \times 1)$ & Relu & -\\

\cmidrule(r){1-6}

Dilated conv & $(m \times h \times 9)$ & $(m \times h \times 16)$ & $(1 \times 3)$ & Relu & 2\\
Dilated conv & $(m \times h \times 16)$ & $(m \times h \times 16)$ & $(1 \times 3)$ & Relu & 2\\
\textit{Corr} layer & $(m \times h \times 16)$ & $(m \times h \times 1)$ & $([m+1] \times 1)$ & Relu & -\\

\cmidrule(r){1-6}

Dilated conv & $(m \times h \times 17)$ & $(m \times h \times 16)$ & $(1 \times 3)$ & Relu & 4\\
Dilated conv & $(m \times h \times 16)$ & $(m \times h \times 16)$ & $(1 \times 3)$ & Relu & 4\\
\textit{Corr} layer & $(m \times h \times 16)$ & $(m \times h \times 1)$ & $([m+1] \times 1)$ & Relu & -\\

\cmidrule(r){1-6}

Causal conv & $(m \times h \times 17)$ & $(m \times 1 \times 16)$ & $(1 \times [h-28])$ & Relu & -\\
$1 \times 1$ conv & $(m \times 1 \times 17)$ & $(m \times 1 \times 1)$ & $(1 \times 1)$ & Softmax & -\\
\bottomrule

\end{tabular}
\label{tbl:Wavenet_structure}
\end{table}

Finally, it is necessary to discuss some connections with the recent work of \cite{zhang2020cost}, where the authors propose an architecture that also takes both sequential and cross-asset dependency into consideration. Their proposed architecture, from a high level perspective, is more complex than ours in that theirs involves two sub-networks, one LSTM and one CNN, whereas ours is built solely on CNN. Our architecture is thus simpler to implement, less susceptible to overfitting,  
and allows for more efficient computation. The most noticeable difference between their design and ours is at the level of the \corrlayer block, where they use a convolution with a $m \times 1$ kernel to extract dependency across assets and apply a standard padding trick to keep the output tensor invariant in size. Their approach suffers from two issues (see Appendix \ref{sec:app:CSPPN} for details): first, the kernel in their design may capture only partial dependency information, and second, most problematically, their design is not asset permutation invariant and thus the performance of their network can be highly sensitive to the ordering of assets. This second issue is further confirmed empirically in section \ref{sec:sensitivityAnal}. 

\begin{remark}
To the best of our knowledge, all the previous literature on permutation invariance (PI) (see for e.g. \cite{deepsets,cai2021,LiWang:2021}) of neural networks have considered a definition that resembles Definition 1 in \cite{cai2021}, which requires that all functions $B \in \mathcal{B}$ are PI, i.e. that for any permutation operator $\sigma$, we have that $\sigma^{-1}\circ B\circ \sigma = B$. Definition \ref{pi} is more flexible as it does not impose the PI property on every functions of the set $\mathcal{B}$ but only on the set as a whole, i.e. if $B\in\mathcal{B}$ then $\sigma^{-1}\circ B\circ \sigma$ is also in $\mathcal{B}$ for all $\sigma$. This distinction is crucial given that \cite{cai2021} themselves observe that correlations cannot be modeled using their version of permutation invariant policy networks "as they are agnostic to identities of entities". For instance, the architecture proposed in figures 2 and 3 violates Cai et al.'s PI definition and therefore does not suffer from this deficiency.
\end{remark}


\section{Experimental results}\label{sec:Experimental}
In this section, we present the results of a series of experiments evaluating the empirical performance of our WaveCorr DRL approach. We start by presenting the experimental set-up. We follow with our main study that evaluates WaveCorr against 
a number of popular benchmarks. We finally shed light on the superior performance of WaveCorr with comparative studies that evaluate the sensitivity of its performance to permutation of the assets, number of assets, size of transaction costs, and (in Appendix \ref{sec:app:maxholding}) maximum holding constraints.  All code is available at \href{https://github.com/saeedmarzban/waveCorr}{https://github.com/saeedmarzban/waveCorr}.

\subsection{Experimental set-up}\label{sec:expSetUp}

\paragraph{Data sets:} We employ three data sets. \textbf{Can-data} includes the daily closing prices of 50 Canadian assets from 01/01/2003 until 01/11/2019 randomly chosen among the 70 companies that were continuously part of the Canadian S\&P/TSX Composite Index during this period. \textbf{US-data} contains 50 randomly picked US assets among the 250 that were part of S\&P500 index during the same period. Finally, \textbf{Covid-data} considered 50 randomly resampled assets from S\&P/TSX Composite Index for period 01/11/2011-01/01/2021 and included open, highest, lowest, and closing daily prices. The Can-data and US-data sets were partitioned into training, validation, and test sets according to the periods 2003-2009/2010-2013/2014-2019 and 2003-2009/2010-2012/2013-2019 respectively, while the Covid-data  was only divided in a training (2012-2018) and testing (2019-2020) periods given that hyper-parameters were reused from the previous two studies. We assume with all datasets a constant comission rate of $c_s=c_p=0.05\%$ in the comparative study, while the sensitivity analysis considers no transaction costs unless specified otherwise.

\paragraph{Benchmarks:} In our main study, we compare the performance of WaveCorr to CS-PPN \citep{zhang2020cost}, EIIE \citep{jiang2017deep}, and the equal weighted portfolio (EW). Note that both CS-PPN and EIIE were adapted to optimize the Sharpe-ratio objective described in section \ref{sec:RL} that exactly accounts for transaction costs.

\paragraph{Hyper-parameter selection:} 
Based on a preliminary unreported investigation, where we explored the influence of different optimizers (namely ADAM, SGD, RMSProp, and SGD with momentum), we concluded that ADAM had the fastest convergence. We also narrowed down a list of reasonable values (see Table \ref{table:hyperparams}) for the following common hyper-parameters: initial learning rate, decay rate, minimum rate, look-back window size $h$, planning horizon $T$. 
For each method, the final choice of hyper-parameter settings was done based on the average annual return achieved on both a 4-fold cross-validation study using Can-data and a 3-fold study with the US-data. The final selection (see Table \ref{tbl:hyperparams}) favored, for each method, a candidate that appeared in the top 5 best performing settings of both data-sets in order to encourage generalization power among similarly performing candidates. Note that in order to decide on the number of epochs, an early stopping criteria was systematically employed.

\paragraph{Metrics:} We evaluate all approaches using out-of-sample data (\quoteIt{test data}). \quoteIt{Annual return} denotes the annualized rate of return for the accumulated portfolio value.  \quoteIt{Annual vol} denotes the prorated standard deviation of daily returns. Trajectory-wise Sharpe ratio (SR) of the log returns,  Maximum drawdown (MDD), i.e. biggest loss from a peak, and average Turnover, i.e. average of the trading volume, are also reported (see \citep{zhang2020cost} for formal definitions). Finaly, we report on the average \quoteIt{daily hit rate} which captures the proportion of days during which the log returns out-performed EW.

\paragraph{Important implementation details:} Exploiting the fact that our SGD step involves exercising the portfolio policy network for $T$ consecutive steps (see equation \eqref{eq:SGD}), a clever implementation was able to reduce WavCorr's training time per episode by a factor of 4.
This was done by replacing the $T$ copies of the portfolio policy network producing $a_0,\,a_2,\,\dots,\,a_{T-1}$, with an equivalent single augmented multi-period portfolio policy network producing all of these actions simultaneously, while making sure that all intermediate calculations are reused as much as possible (see Appendix \ref{sec:app:input} for details). 
We also implement our stochastic gradient descent approach by updating, after each episode $k$, the initial state distribution $F$ to reflect the latest policy $\mu_{\theta_k}$. This is done in order for the final policy to be better adapted to the conditions encountered when the portfolio policy network is applied on a longer horizon than $T$.

\subsection{Comparative Evaluation of WaveCorr}\label{sec:compEval}

In this set of experiments the performances of WaveCorr, CS-PPN, EIIE, and EW are compared for a set of 10 experiments (with random reinitialization of NN parameters) on the three datasets. 
The average and standard deviations of each performance metric are presented in Table \ref{tbl:TCMain} while Figure \ref{fig:mainexperiment} (in the Appendix) presents the average out-of-sample portfolio value trajectories. The main takeaway from the table is that WaveCorr significantly outperforms the three benchmarks on all data sets, achieving an absolute improvement in average yearly returns of 3\% to 25\% compared to the best alternative. It also dominates CS-PPN and EIIE in terms of Sharpe ratio, maximum drawdown, daily hit rate, and turnover. EW does appear to be causing less volatility in the US-data, which leads to a slightly improved SR. Another important observation consists in the variance of these metrics over the 10 experiments. Once again WaveCorr comes out as being generally more reliable than the two other DRL benchmarks in the Can-data, while EIIE appears to be more reliable in the US-data sacrificing average performance. 
Overall, the impressive performance of WaveCorr  seems to support our claim that our new architecture allows for a better identification of the cross-asset dependencies. 
In conditions of market crisis (i.e. the Covid-data), we finally observe that WaveCorr exposes the investors to much lower short term losses, with an MDD of only 31\% compared to more than twice as much for CS-PPN and EIIE, which reflects of a more effective hedging strategy. 


\begin{table}[h]
\caption{The average (and standard deviation) performances using three data sets.}
\label{tbl:TCMain}
\centering
\begin{tabular}{lcccccc}
\toprule
Method &	Annual return   	&	Annual vol      	&	SR	&	MDD    	&	Daily hit rate  	&	Turnover        	\\
\midrule
\multicolumn{7}{c}{Can-data}\\
\midrule
WaveCorr	&	27\%	(3\%)	&	16\%	(1\%)	&	1.73	(0.25) &	16\%	(2\%)&	52\% (1\%)	&	0.32	(0.01)\\

CS-PPN	&	21\% (4\%)	&	19\% (2\%)	&	1.14 (0.34)	&	17\% (4\%)	&	51\% (1\%)	&	0.38 (0.05)	\\

EIIE	&	-1\% (8\%)	&	29\% (4\%)	&	-0.01 (0.28)	&	55\% (9\%)	&	47\% (1\%)	&	0.64 (0.08)	\\

EW	&	\;4\% (0\%) &	14\% (0\%)	&	0.31 (0.00)\;	&	36\% (0\%)	&	-	&	0.00 (0.00)	\\
\midrule
\multicolumn{7}{c}{US-data}\\
\midrule
WaveCorr	&	19\% (2\%)	&	16\% (2\%)	&	1.17 (0.20)	&	20\% (4\%)	&	50\% (1\%)	&	0.11 (0.02)	\\

CS-PPN	&	14\% (2\%)	&	15\% (2\%)	&	0.94 (0.17)	&	22\% (6\%)	&	49\% (1\%)	&	0.15 (0.08)	\\

EIIE	&	16\% (1\%)	&	15\% (0\%)	&	1.09 (0.06)	&	20\% (1\%)	&	50\% (0\%)	&	0.17 (0.02)	\\

EW	&	15\% (0\%)	&	13\% (0\%)	&	1.18 (0.00)	&	18\% (0\%)	&	-	&	0.00 (0.00)	\\
\midrule
\multicolumn{7}{c}{Covid-data}\\
\midrule
WaveCorr	&	56\% (13\%)	&	26\% (5\%)	&	2.16 (0.50)	&	31\% (9\%)	&	51\% (2\%)	&	0.19 (0.05)	\\

CS-PPN	&	31\% (27\%)	&	51\% (6\%)	&	0.60 (0.48)	&	67\% (7\%)	&	50\% (2\%)	&	0.3 (0.09)	\\

EIIE	&	11\% (30\%)	&	76\% (17\%)	&	0.20 (0.43)	&	77\% (13\%)	&	46\% (2\%)	&	0.76 (0.27)	\\

EW	&	27\% (0\%)	&	29\% (0\%)	&	0.93(0.00)	&	47\% (0\%)	&	-	&	0.01 (0.00)	\\
\bottomrule

\end{tabular}
\end{table}

\removed{
\begin{figure}[h!]
	\centering
	\begin{subfigure}{\mynewwidth\textwidth}
		\includegraphics[scale=\mynewscale]{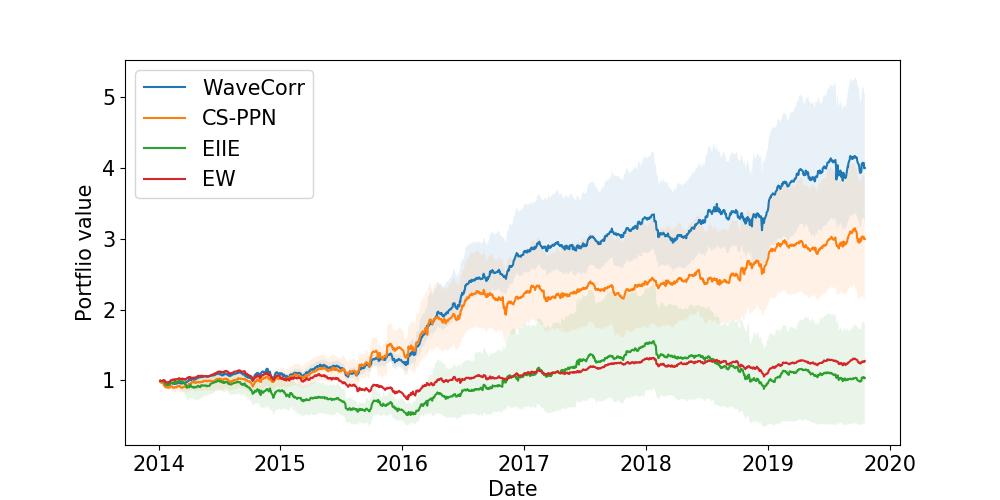}
		\caption{Can-data}
	\end{subfigure}
	\begin{subfigure}{\mynewwidth\textwidth}
		\includegraphics[scale=\mynewscale]{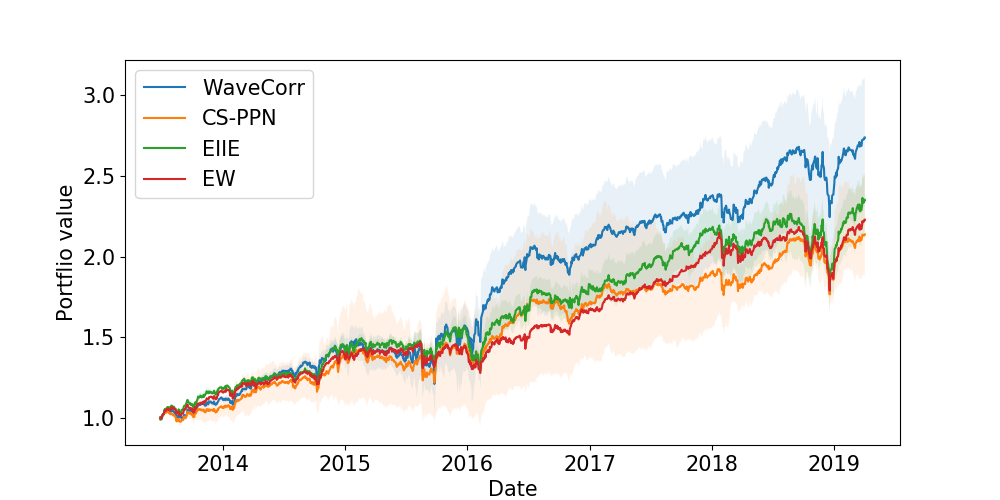}
		\caption{US-data}
	\end{subfigure}
	
 	\begin{subfigure}{\mynewwidth\textwidth}
 		\includegraphics[scale=\mynewscale]{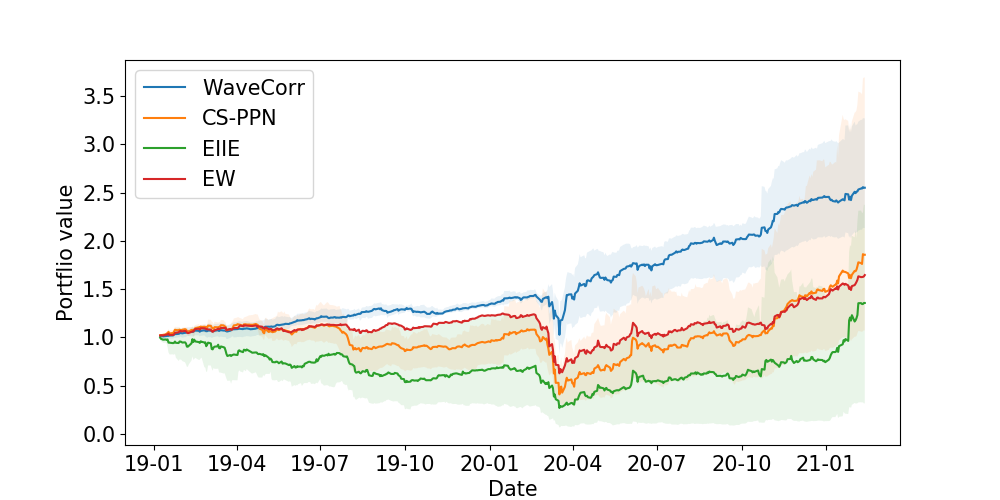}
 		\caption{Covid-data}
 	\end{subfigure}
	
	\caption{Comparison of the wealth accumulated (a.k.a. portfolio value) by WaveCorr, CS-PPN, EIIE, and EW over Can-data, US-data, and Covid-data.
	}
	\label{fig:mainexperiment}
	
\end{figure}}


													
													



\subsection{Sensitivity Analysis}\label{sec:sensitivityAnal}

\paragraph{Sensitivity to permutation of the assets:}

In this set of experiment, we are interested in measuring the effect of asset permutation on the performance of WaveCorr and CS-PPN. Specifically, each experiment now consists in resampling a permutation of the 50 stocks instead of the initial parameters of the neural networks. The results are summarized in Table \ref{tbl:wavecorr_compare_costsens:permute} and illustrated in Figure \ref{fig:wavecorr_compare_costsens_ca_permutation}. We observe that the learning curves and performance of CS-PPN are significantly affected by asset permutation compared to WaveCorr. In particular, one sees that the standard deviation of annual return is reduced by more than a factor of about 5 with WaveCorr. We believe this is entirely attributable to the new structure of the \corrlayer in the portfolio policy network.

\begin{table}[h]
\caption{The average (and standard dev.) performances over random asset permutation in Can-data.}
\label{tbl:wavecorr_compare_costsens:permute}
\centering
\begin{tabular}{lccccccc}
\toprule
&	Annual return   	&	Annual vol      	&	SR	&	MDD    	&	Daily hit rate  	&	Turnover        	\\
\midrule
WaveCorr	&	48\% (1\%)	&	15\% (1\%)	&	3.15 (0.19)	&	14\% (3\%)	&	56\% (0\%)	&	0.48 (0.01)	\\

CS-PPN	&	35\% (5\%)	&	18\% (1\%)	&	2.00 (0.37)	&	22\% (4\%)	&	54\% (1\%)	&	0.54 (0.03)	\\
\bottomrule
\end{tabular}
\end{table}

\begin{figure}[h]
	\centering
\removed{	\begin{subfigure}{\mynewwidth\textwidth}
		\includegraphics[scale=\mynewscale]{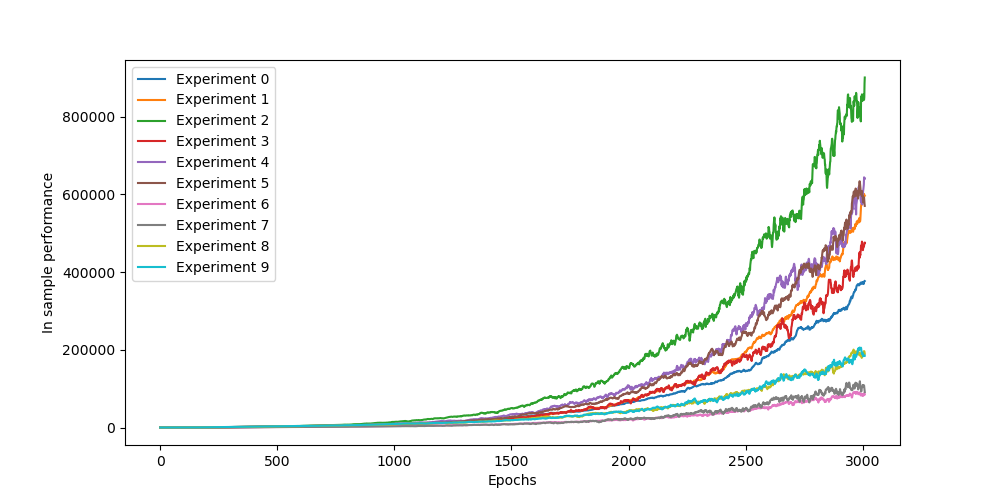}
		\caption{In sample learning curve, CS-PPN}
	\end{subfigure}
	\begin{subfigure}{\mynewwidth\textwidth}
		\includegraphics[scale=\mynewscale]{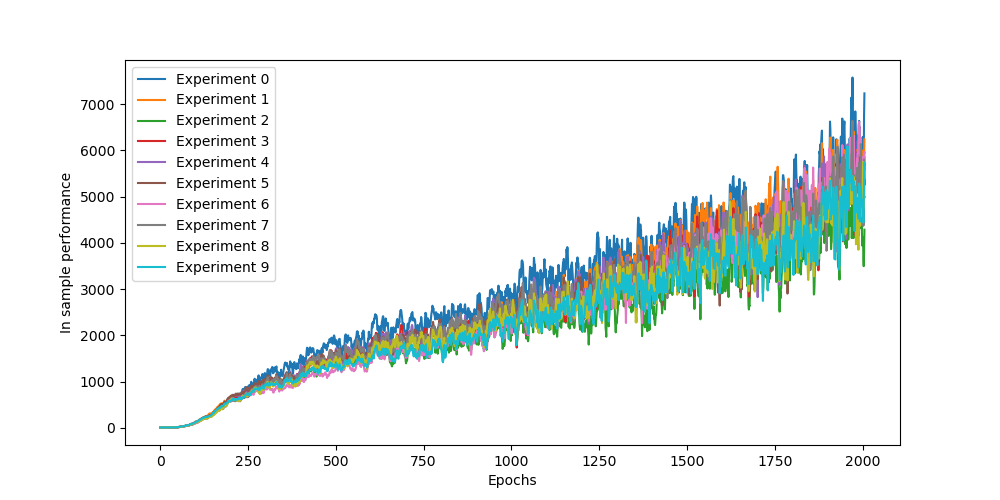}
		\caption{In sample learning curve, WaveCorr}
	\end{subfigure}
	
	\begin{subfigure}{\mynewwidth\textwidth}
		\includegraphics[scale=\mynewscale]{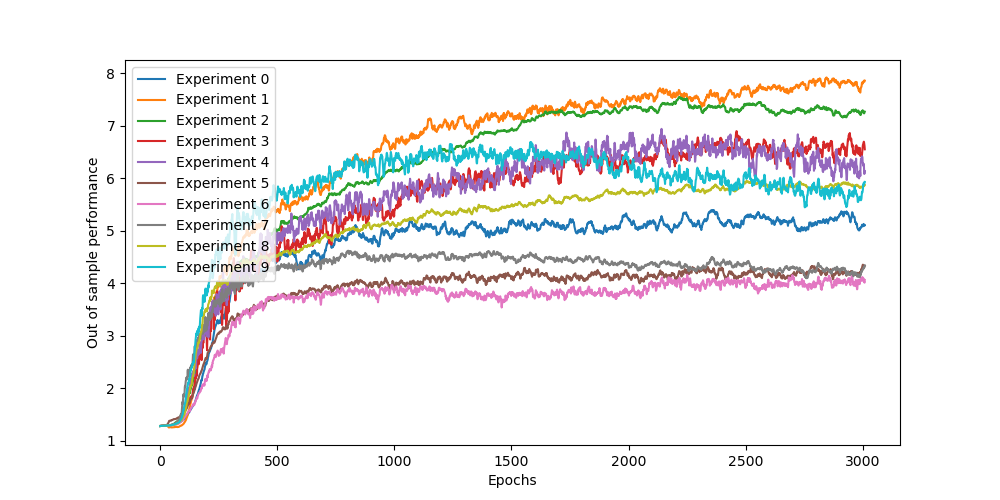}
		\caption{Out of sample learning curve, CS-PPN}
	\end{subfigure}
	\begin{subfigure}{\mynewwidth\textwidth}
		\includegraphics[scale=\mynewscale]{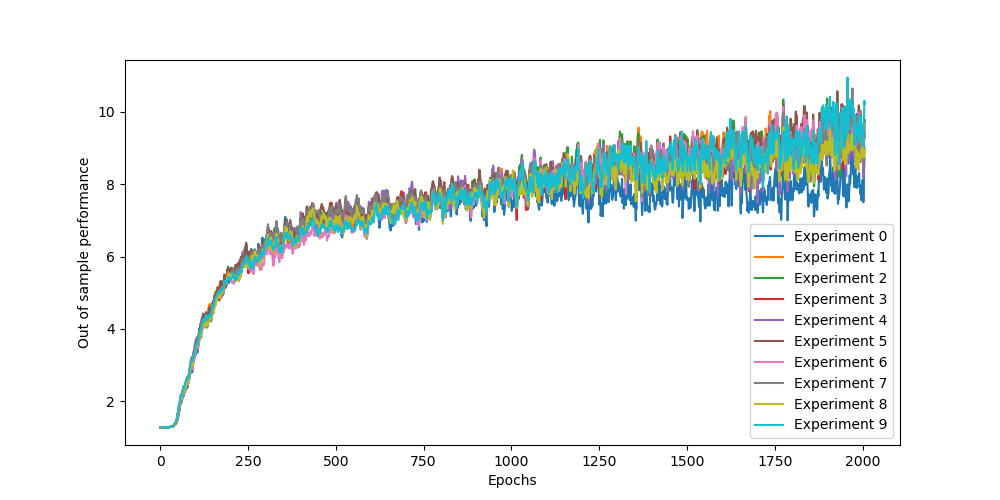}
		\caption{Out of sample learning curve, WaveCorr}
	\end{subfigure}}

	\begin{subfigure}{\mynewwidth\textwidth}
		\includegraphics[scale=\mynewscale]{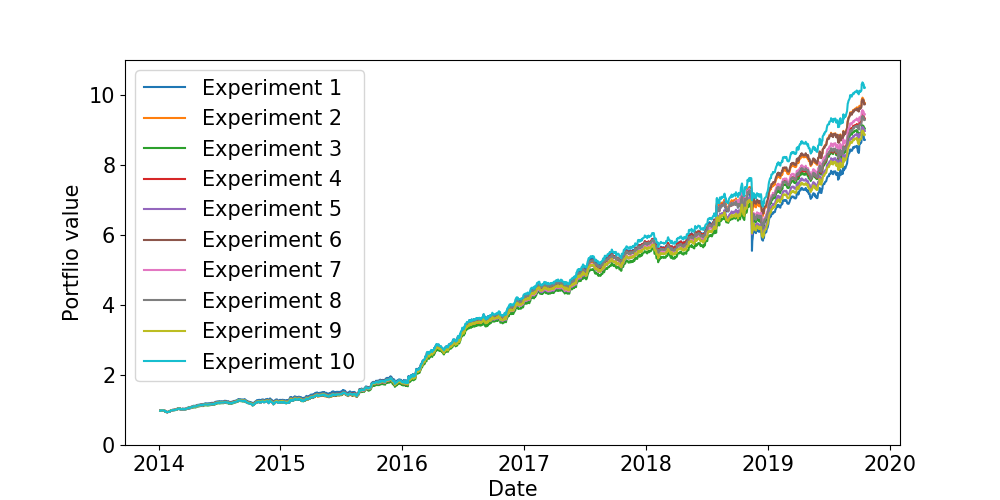}
		\caption{WaveCorr}
	\end{subfigure}
	\begin{subfigure}{\mynewwidth\textwidth}
		\includegraphics[scale=\mynewscale]{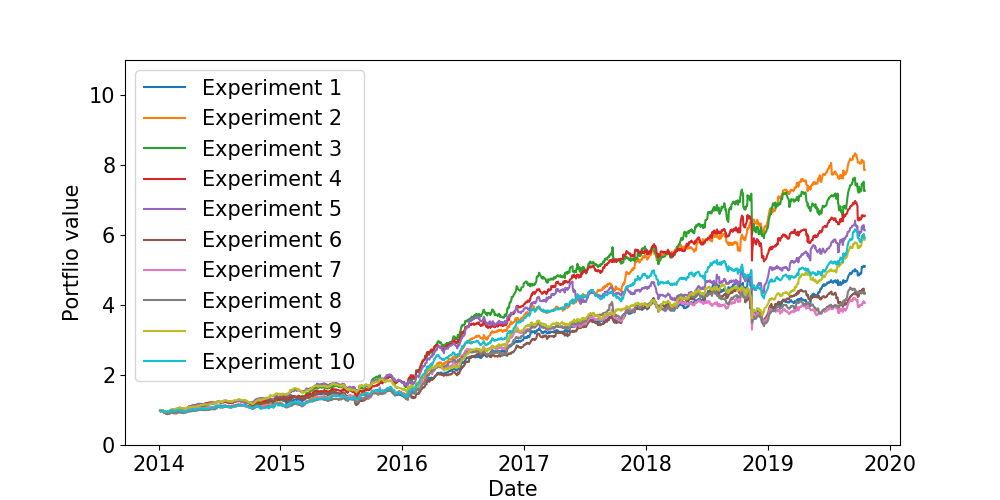}
		\caption{CS-PPN}
	\end{subfigure}
	
	\caption{Comparison of the wealth accumulated by WaveCorr and CS-PPN under random initial permutation of assets on Can-data's test set. 
	}
	\label{fig:wavecorr_compare_costsens_ca_permutation}
	\vspace{-0.2cm}
\end{figure}

	
	
	
	

\paragraph{Sensitivity to number of assets:}

In this set of experiments, we measure the effect of varying the number of assets on the performance of WaveCorr and CS-PPN. We therefore run 10 experiments (randomly resampling initial NN parameters) with growing subsets of 30, 40, and 50 assets from Can-data. Results are summarized in Table \ref{tbl:increaseStocksCA} and illustrated in Figure \ref{fig:wavecorr_compare_costsens_increase} (in Appendix). While having access to more assets should in theory be beneficial for the portfolio performance, we observe that it is not necessarily the case for CS-PPN. On the other hand, as the number of assets increase, a significant improvement, with respect to all metrics, is achieved by WaveCorr. This evidence points to a better use of the correlation information in the data by WaveCorr. 

\begin{table}[h]
	\vspace{-0.2cm}
\caption{The average (and std. dev.) performances as a function of the number of assets in Can-data.}
\label{tbl:increaseStocksCA}
\centering
\begin{tabular}{ccccccc}
\toprule
\# of stocks &	Annual return   	&	Annual vol      	&	SR	&	MDD    	&	Daily hit rate  	&	Turnover        	\\
\midrule
\multicolumn{7}{c}{WaveCorr}\\
\midrule
30	&	37.7\% (4\%)	&	19\% (1\%)	&	2.02 (0.27)	&	22\% (3\%)	&	55\% (1\%)	&	0.39 (0.02)	\\

40	&	38.5\% (4\%)	&	21\% (1\%)	&	1.81 (0.17)	&	23\% (2\%)	&	55\% (1\%)	&	0.44 (0.04)	\\

50	&	43.0\% (5\%)	&	17\% (2\%)	&	2.57 (0.52)	&	20\% (6\%)	&	55\% (1\%)	&	0.43 (0.02)	\\

\midrule
\multicolumn{7}{c}{CS-PPN}\\
\midrule
30	&	30.3\% (3\%)	&	17\% (1\%)	&	1.80 (0.20)	&	21\% (4\%)	&	53\% (1\%)	&	0.42 (0.04)	\\

40	&	29.8\% (7\%)	&	17\% (2\%)	&	1.70 (0.34)	&	22\% (3\%)	&	53\% (1\%)	&	0.41 (0.09)	\\

50	&	32.2\% (4\%)	&	16\% (1\%)	&	2.07 (0.28)	&	18\% (3\%)	&	52\% (1\%)	&	0.43 (0.05)	\\

\bottomrule

\end{tabular}
\end{table}

\paragraph{Sensitivity to commission rate:}
Table \ref{tbl:TC} presents how the performances of WaveCorr and CS-PPN are affected by the magnitude of the commission rate, ranging among 0\%, 0.05\%, and 0.1\%. One can first recognize that the two methods appear to have good control on turnover as the commission rate is increased. Nevertheless, one can confirm from this table the significantly superior performance of WaveCorr prevails under all level of commission rate. 

\begin{table}[h]
	\vspace{-0.2cm}
\caption{The average (and std. dev.) performances as a function of commission rate (CR) in Can-data.}
\label{tbl:TC}
\centering
\begin{tabular}{lcccccc}
\toprule
Method&	Annual return   	&	Annual vol      	&	SR	&	MDD    	&	Daily hit rate  	&	Turnover        	\\
\midrule
\multicolumn{7}{c}{CR = 0}\\
\midrule

WaveCorr	&	42\% (3\%)	&	15\% (0\%)	&	2.77 (0.20)	&	13\% (1\%)	&	55\% (1\%)	&	0.44 (0.02)	\\

CS-PPN	&	35\% (4\%)	&	17\% (1\%)	&	2.04 (0.27)	&	14\% (3\%)	&	53\% (1\%)	&	0.47 (0.05)	\\

\midrule
\multicolumn{7}{c}{CR = 0.05\%}\\
\midrule

WaveCorr	&	27\% (3\%)	&	16\% (1\%)	&	1.73 (0.25)	&	15\% (2\%)	&	52\% (1\%)	&	0.32 (0.01)	\\

CS-PPN	&	21\% (4\%)	&	19\% (2\%)	&	1.14 (0.34)	&	17\% (4\%)	&	51\% (1\%)	&	0.38 (0.05)	\\

\midrule
\multicolumn{7}{c}{CR = 0.1\%}\\
\midrule

WaveCorr	&	19\% (2\%)	&	15\% (1\%)	&	1.34 (0.16)	&	16\% (2\%)	&	50\% (1\%)	&	0.22 (0.01)	\\
CS-PPN	&	14\% (7\%)	&	17\% (3\%)	&	0.92 (0.50)	&	19\% (8\%)	&	50\% (1\%)	&	0.22 (0.09)	\\

\bottomrule

\end{tabular}
\end{table}

\section{Conclusion}\label{sec:Conclusion}

This paper presented a new architecture for portfolio management that is built upon WaveNet  \citep{oord2016wavenet}, which uses dilated causal convolutions at its core, and a new design of correlation block that can process and extract cross-asset information. We showed that, despite being parsimoniously parameterized, WaveCorr can satisfy the property of asset permutation invariance,   whereas a naive extension of CNN, such as in the recent works of \cite{zhang2020cost}, does not. 
The API property is both appealing from a practical, given that it implies that the investor does not need to worry about how he/she indexes the different assets, and empirical point of views, given the empirical evidence that it leads to improved stability of the network's performance. As a side product of our analysis, the results presented in Appendix \ref{sec:app:proofSecArch} lay important foundations for analysing the API property in a larger range of network architectures.
In the numerical section, we tested the performance of WaveCorr using data from both Canadian (TSX) and American (S\&P 500) stock markets. The experiments demonstrate that WaveCorr consistently outperforms our benchmarks under a number of variations of the model: including the number of available assets, the size of transaction costs, etc.


\section*{Acknowledgement}
The authors gratefully acknowledge the financial support from the Canadian Natural Sciences and Engineering Research Council [Grants RGPIN-2016-05208 and RGPIN-2014-05602], Mitacs [Grant IT15577], Compute Canada, Evovest, and the NSERC-CREATE Program on Machine Learning in Quantitative Finance and Business Analytics.

\oldText{
\subsection{The relation between in and out of sample performance\EDcomments{This should be measured and discussed in the MAIN study}}

In the third set of experiments, we investigate the correlation between the in and out of sample performance of the models. In supervised learning there is a tight relation between these two levels of performance, as if a model is not overfitting the out of sample performance would not be much different from the in sample. However, when it comes to RL, this relation is not as obvious. It may happen that with significantly different in sample performance levels, we observe similar out of samples. Having in and out of sample results that are close to each other could be beneficial as we can implement the average of a set of our runs with the highest in sample performance instead of averaging over all of them when implementing the policies and obtain higher performance. However, if there is no clear relation beween the two, there is no clue which one of the runs has resulted in better out of sample performance. In this section we provide numerical results to demonstrate that WaveCorr is out performing the cost-sensitve model in providing closer in and out of sample results. To do so, we run 20 experiments by using each model, cost-sensitive and the WaveCorr, and then compute the cross correlation of the in and out of sample results. Our results show a significant positive correlation for WaveCorr and a negative Correlation for the cost-sensitive model. More espicifically, we computed a correlation of 0.73 for the WaveCorr and -0.56 for the cost-sensitive model, which shows the significan improvement the WaveCorr model is providing.

\begin{figure}[h]
	\centering
	\begin{subfigure}{\mynewwidth\textwidth}
		\includegraphics[scale=\mynewscale]{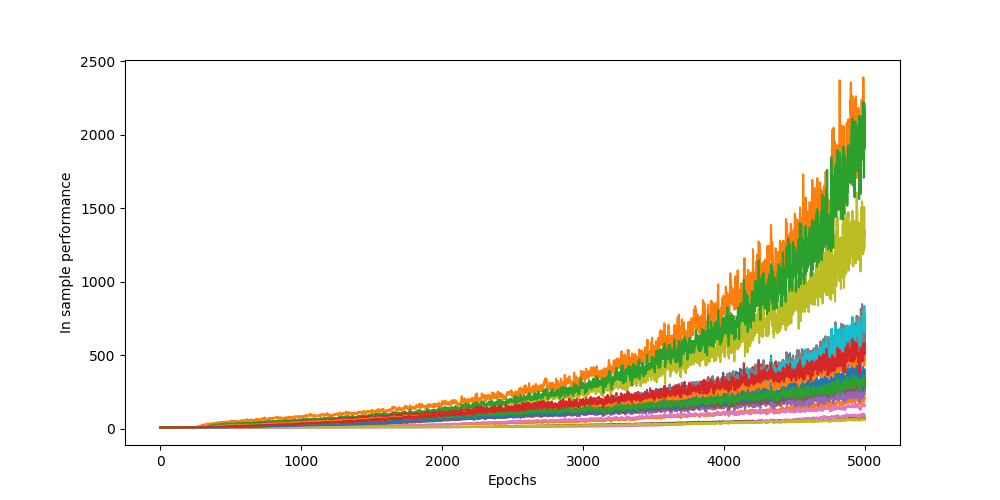}
		\caption{In sample, Cost-sensitive}
	\end{subfigure}
	\begin{subfigure}{\mynewwidth\textwidth}
		\includegraphics[scale=\mynewscale]{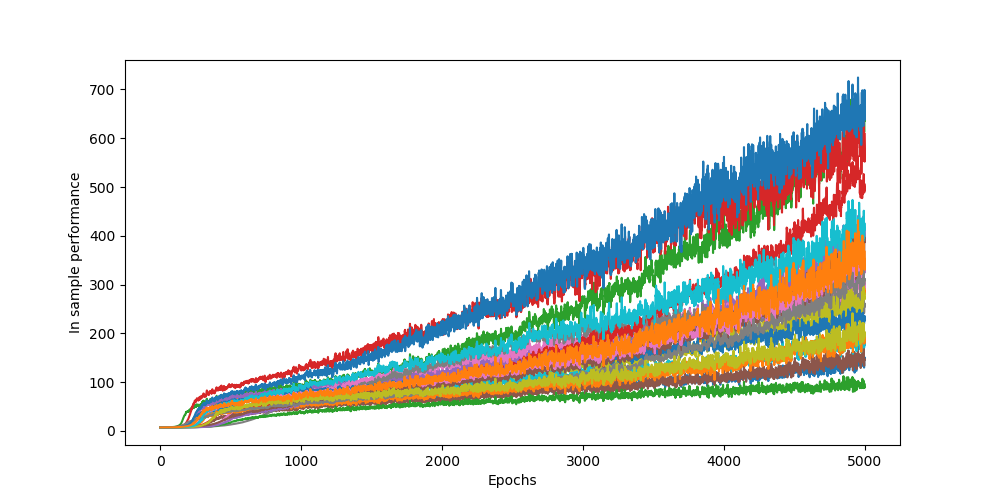}
		\caption{In sample, WaveCorr}
	\end{subfigure}
	
	\begin{subfigure}{\mynewwidth\textwidth}
		\includegraphics[scale=\mynewscale]{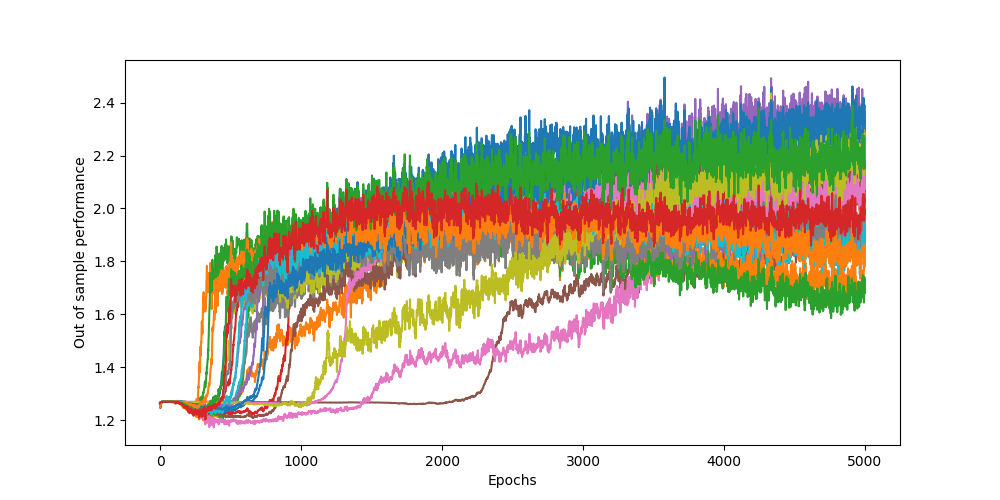}
		\caption{Out of sample, Cost-sensitive}
	\end{subfigure}
	\begin{subfigure}{\mynewwidth\textwidth}
		\includegraphics[scale=\mynewscale]{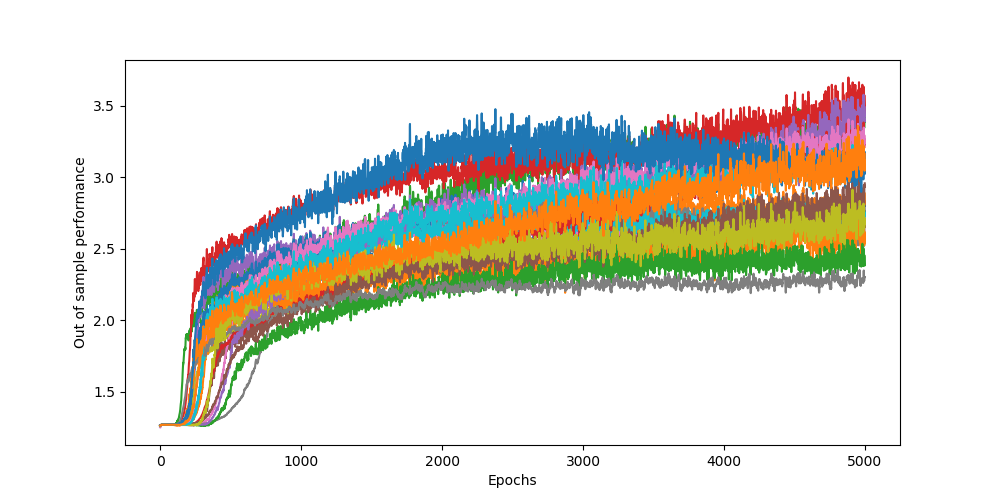}
		\caption{Out of sample, WaveCorr}
	\end{subfigure}
	
	\begin{subfigure}{\mynewwidth\textwidth}
		\includegraphics[scale=\mynewscale]{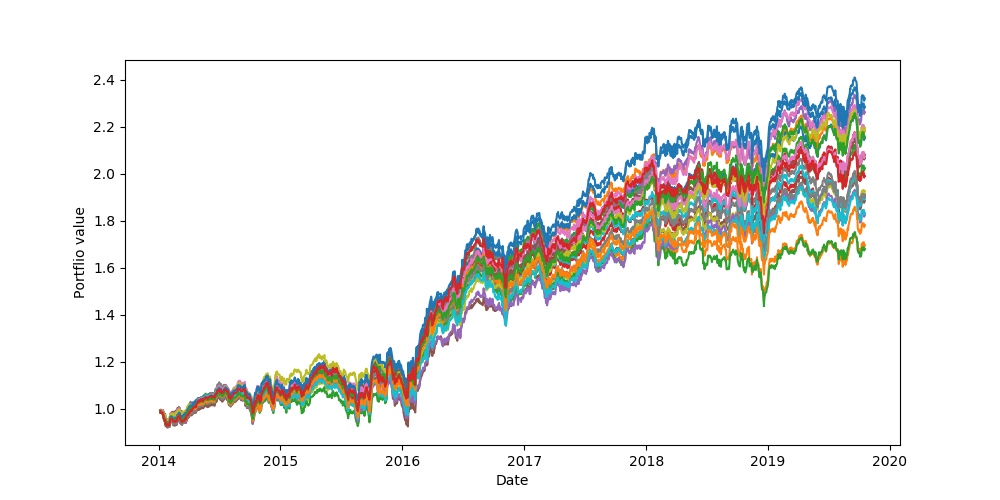}
		\caption{Portfolio value, Cost-sensitive}
	\end{subfigure}
	\begin{subfigure}{\mynewwidth\textwidth}
		\includegraphics[scale=\mynewscale]{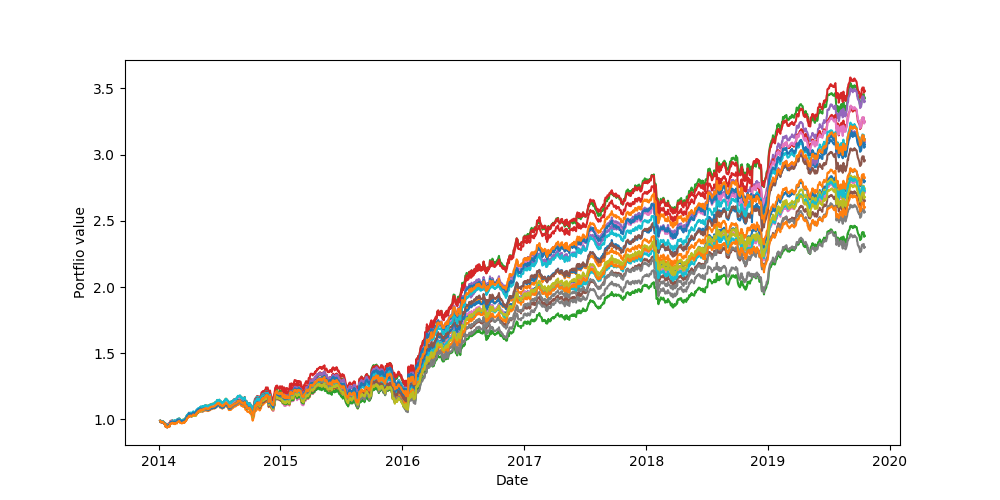}
		\caption{Portfolio value, WaveCorr}
	\end{subfigure}
	
	\caption{Comparing the correlation between the in and out of sample performance of WaveCorr and Cost-senstive models}
	
\end{figure}
}

\oldText{
\subsection{COVID-19 data}
During the covid-19 time, the market experienced significant drawdowns. In this section we test our model compared to the Cost-sensitive model by testing them during Jan-2019 and Jan-2021. The results are shown in Figure \ref{fig:covid} and Table \ref{tbl:covid}.}

\appendix
\section{Appendix}

\renewcommand\thesection{\Alph{section}}
\renewcommand{\theequation}{A.\arabic{equation}}
\renewcommand{\thefigure}{A.\arabic{figure}}
\renewcommand{\thetable}{A.\arabic{table}}

This appendix is organized as follows. Section \ref{app:bissect} demonstrates a claim made in section \ref{sec:PM_problem} regarding the fact that the solution of $\nu_t = f(\nu_t,\PW_t',\PW_t)$ can be obtained using a bisection method. Section \ref{sec:app:proofSecArch} presents proofs to the two propositions in section \ref{sec:architecture}. Section \ref{sec:app:CSPPN} presents further details on the correlation layer in \cite{zhang2020cost} and its two deficiencies.  Section \ref{sec:app:input} presents further details on the augmented policy network architecture used to accelerate training. Section \ref{sec:app:hyperParams} presents our hyper-parameter ranges and final selection. Finally, section \ref{sec:app:addRes} presents a set of additional results.
    
\subsection{Solving $\nu = f(\nu)$}\label{app:bissect}

In order to apply the bisection method to solve $\nu=f(\nu)$, we will make use of the following proposition.
\begin{proposition}\label{thm:bissectThm}
For any $0<c_s<1$ and $0<c_p<1$, the function $g(\nu):=\nu-f(\nu)$ is strictly increasing on $[0,\,1]$ with $g(0)<0$ and $g(1)>0$.
\end{proposition}

\begin{proof}
Recalling that $f(\nu,\PW',\PW):=1- c_s \sum_{i=1}^{m} (\pw'^{i} - \nu \pw^i)^+ - c_p \sum_{i=1}^{m} (\nu \pw^i - \pw'^{i})^+$, we first obtain the two bounds at $g(0)$ and $g(1)$ as follows:
\[g(0)=0-\left(1- c_s \sum_{i=1}^{m} (\pw'^{i})^+ - c_p \sum_{i=1}^{m} (- \pw'^{i})^+\right)=0-1+c_s<0\,,\]
since $c_s<1$, and
\[g(1)=1-\left(1- c_s \sum_{i=1}^{m} (\pw'^{i} - \pw^i)^+ - c_p \sum_{i=1}^{m} (\pw^i - \pw'^{i})^+\right)\geq \min(c_s,c_p)\|\PW'^{i} - \PW^i\|_1>0\,,\]
since $\min(c_s,c_p)>0$.
We can further establish the convexity of $g(\nu)$, given that it is the sum of convex functions. A careful analysis reveals that $g(\nu)$ is supported at $0$ by the plane
\[g(\nu)\geq g(0)+\nu\left(1-c_s+c_p\sum_{i=1}^m \1\{\pw^{\prime i}=0\}\pw^i\right)\,,\]
where $\1\{A\}$ is the indicator function that returns $1$ if $A$ is true, and $0$ otherwise.
Hence, by convexity of $g(\nu)$, the fact that this supporting plane is strictly increasing implies that $g(\nu)$ is strictly increasing for all $\nu\geq 0$.
\end{proof}

Given Proposition \ref{thm:bissectThm}, we can conclude that a bisection method can be used to find the root of $g(\nu)$, which effectively solves $\nu=f(\nu)$.


\subsection{Proofs of Section \ref{sec:architecture}}\label{sec:app:proofSecArch}

\EDmodified{We start this section with a lemma that will simplify some of our later derivations.
\begin{lemma}\label{thm:app:subsetAPI}
A block capturing a set of functions $\mathcal{B}\subseteq \{B:\mathbb{R}^{m \times h \times d}\rightarrow \mathbb{R}^{m \times h' \times d'}\}$ is asset permutation invariant if and only if given any permutation operator $\sigma$, we have that $\{\sigma^{-1}\circcomp B\circcomp \sigma: B\in\mathcal{B}\} \supseteq \mathcal{B}$.
\end{lemma}
\begin{proof}
The \quoteIt{only if} follows straightforwardly from the fact that equality between two sets implies that each set is a subset of the other.

Regarding the \quoteIt{if} part, we start with the assumption that 
\[\forall \sigma, \{\sigma^{-1}\circcomp B\circcomp \sigma: B\in\mathcal{B}\} \supseteq \mathcal{B}\,.\]
Next, we follow with the fact that for all permutation operator $\sigma$:
\begin{align*}
\{\sigma^{-1}\circcomp B\circcomp \sigma: B\in\mathcal{B}\}&\subseteq \{\sigma^{-1}\circcomp B\circcomp \sigma: B\in\{\sigma\circcomp B'\circcomp \sigma^{-1}: B'\in\mathcal{B}\}\}\\
&=\{\sigma^{-1}\circcomp \sigma\circcomp B'\circcomp \sigma^{-1}\circcomp \sigma: B'\in\mathcal{B}\}=\mathcal{B}\,,
\end{align*}
where we assumed for simplicity of exposition that $h=h'$ and $d=d'$, and exploited the fact that $\sigma^{-1}$ is also a permutation operator.
\end{proof}}

\subsubsection{Proof of Proposition \ref{thm:corrlayer}}

We first clarify that the correlation layer is associated with the following set of functions (see Procedure 1):
\[\mathcal{B}:=\{B_{w,b}:w\in\myRe^{(m+1)\times d},b\in\myRe\}\]
where
$$B_{w,b}(\mathcal{T})[i,:,1]:=\left( {\cal T}[i,:,:]\circprod  (\vec{1} w_0^\top) + \sum_{j=1}^m {\cal T}[j,:,:] \circprod  (\vec{1} w_j^\top)\right) \vec{1}  + b, \;\; \forall\,i=1,...,m,$$
with $\circprod$ denoting the Hadamard (element-wise) product. 

Let $\sigma$ (associated with the bijection $\pi$) be an asset permutation operator. For any correlation layer function $B_{w,b}\in\mathcal{B}$, one can construct a new set of parameters $w_0':=w_0$, $w_j':=w_{\pi(j)}$, for all $j=1,\dots,m$, and $b':=b$ such that for all input tensor $\mathcal{T}$, we have that for all $i$:
\begin{align}
    B_{w',b'}(\sigma(\mathcal{T}))[i,:,1]&= \left(\sigma({\cal T})[i,:,:]\circprod(\vec{1}w_{0}^{\top})+\sum_{j=1}^{m}\sigma({\cal T})[j,:,:]\circprod(\vec{1}w_{\pi(j)}^{\top})\right)\vec{1} + b\\
&=  \left({\cal T}[\pi(i),:,:]\circprod(\vec{1}w_{0}^{\top})+\sum_{j=1}^{m}{\cal T}[\pi(j),:,:]\circprod(\vec{1}w_{\pi(j)}^{\top})\right)\vec{1} + b\\
&=  \left({\cal T}[\pi(i),:,:]\circprod(\vec{1}w_{0}^{\top})+\sum_{j'=1}^{m}{\cal T}[j',:,:]\circprod(\vec{1}w_{j'}^{\top})\right)\vec{1} + b.
\end{align}
Hence,
\[\sigma^{-1}(B_{w',b'}(\sigma(\mathcal{T})))[i,:,1]=\left({\cal T}[i,:,:]\circprod(\vec{1}w_{0}^{\top})+\sum_{j=1}^{m}{\cal T}[j,:,:]\circprod(\vec{1}w_{j}^{\top})\right)\vec{1} + b= B_{w,b}(\mathcal{T})[i,:,1]\,.\]
\EDmodified{We can therefore conclude that $\{\sigma^{-1}\circcomp B\circcomp \sigma: B\in\mathcal{B}\} \supseteq \mathcal{B}$. Based on Lemma \ref{thm:app:subsetAPI}, we conclude that $\mathcal{B}$ is asset permutation invariant.}

\subsubsection{Proof of Proposition \ref{thm:permInv}}\label{sec:app:proofPI}
To prove Proposition \ref{thm:permInv}, we demonstrate that all blocks used in the WaveCorr architecture are asset permutation invariant (Steps 1 to 3). We then show that asset permutation invariance is preserved under composition (Step 4). Finally, we can conclude in Step 5 that WaveCorr is asset permutation invariant.

\paragraph{Step 1 - Dilated convolution, Causal convolution, Sum, and $1 \times 1$ convolution are asset permutation invariant:}
The functional class of a dilated convolution, a causal convolution, a sum, and a $1 \times 1$ convolution block all have the form:
\[\mathcal{B}:=\{B_g: g\in\mathcal{G}\},\]
where
\[B_g(\mathcal{T})[i,:,:]:=g(\mathcal{T}(i,:,:)), \;\; \forall\,i=1,...,m,\]
for some set of functions $\mathcal{G}\subseteq \{G:\myRe^{1\times h\times d}\rightarrow\myRe^{1\times h\times d'}\}$. In particular, in the case of dilated, causal, and $1\times 1$ convolutions, this property follows from the use of $1\times 3$, $1\times[h-28]$, and $1\times 1$ kernels respectively. Hence, for any $g\in\mathcal{G}$, we have that:
\[\sigma^{-1}(B_g(\sigma(\mathcal{T})))=\sigma^{-1}(\sigma(B_g(\mathcal{T})))=B_g(\mathcal{T})\,,\]
\EDmodified{which implies that $\{\sigma^{-1}\circcomp B\circcomp\sigma: B\in\mathcal{B}\} = \mathcal{B}$.}

\paragraph{Step 2 - Relu and dropout are asset permutation invariant:}
We first clarify that Relu and dropout on a tensor in $\myRe^{m\times h\times d}$ are singleton sets of functions:
\[\mathcal{B}:=\{B_g\}\]
where $g:\myRe\rightarrow\myRe$ and $B_g(\mathcal{T})[i,j,k]:=g(\mathcal{T}[i,j,k])$. In particular, in the case of Relu, we have:
$$B_g(\mathcal{T})[i,j,k]:=\max(0,\mathcal{T}[i,j,k])\,,$$
while, for dropout we have: 
$$B_g(\mathcal{T})[i,j,k]:=\mathcal{T}[i,j,k]\,,$$
since a dropout block acts as a feed through operator. 
Hence, we naturally have that:
\[\sigma^{-1}(B_g(\sigma(\mathcal{T})))=\sigma^{-1}(\sigma(B_g(\mathcal{T})))=B_g(\mathcal{T})\,,\]
\EDmodified{which again implies that $\{\sigma^{-1}\circcomp B\circcomp\sigma: B\in\mathcal{B}\} = \mathcal{B}$.}

\paragraph{Step 3 - Softmax is asset permutation invariant:}

We first clarify that softmax on a vector in $\myRe^{m\times h\times 1}$ is a singleton set of functions:
\[\mathcal{B}:=\{B\}\]
where
$$B(\mathcal{T})[i,j,1]:=\frac{\exp(\mathcal{T}[i,j,1])}{\sum_{i'=1}^m \exp(\mathcal{T}[i',j,1])}\,.$$
Hence, we have that:
\[B(\sigma(\mathcal{T}))[i,j,1]:=\frac{\exp(\mathcal{T}[\pi(i),j,1])}{\sum_{i'=1}^m \exp(\mathcal{T}[\pi(i'),j,1])}=\frac{\exp(\mathcal{T}[\pi(i),j,1])}{\sum_{i'=1}^m \exp(\mathcal{T}[i',j,1])}\,.\]
This allows us to conclude that:
\[\sigma^{-1}(B(\sigma(\mathcal{T})))[i,j,1]=B(\mathcal{T})[i,j,1]\,.\]
\EDmodified{Hence, we conclude that $\{\sigma^{-1}\circcomp B\circcomp \sigma: B\in\mathcal{B}\} = \mathcal{B}$.}

\paragraph{Step 4 - Asset permutation invariance is preserved under composition:}
Given two asset permutation invariant blocks representing the set of functions $\mathcal{B}_1$ and $\mathcal{B}_2$, one can define the composition block as:
\[\mathcal{B}:=\{B_1\circcomp B_2: B_1\in\mathcal{B}_1,B_2\in\mathcal{B}_2\}\,.\]
We have that for all $B_1\in \mathcal{B}_1$ and $B_2\in\mathcal{B}_2$:
\begin{align*}
B &= B_1\circcomp B_2\\
&= (\sigma^{-1}\circcomp B_1'\circcomp \sigma)\circcomp(\sigma^{-1}\circcomp B_2' \circcomp \sigma)\\
&=\sigma^{-1}\circcomp B_1'\circcomp B_2' \circcomp \sigma\\
&=\sigma^{-1}\circcomp B' \circcomp \sigma\,,
\end{align*}
where $B_1'\in\mathcal{B}_1$ and $B_2'\in\mathcal{B}_2$ come from the definition of asset permutation invariance, and where $ B' :=B_1'\circcomp B_2'\in\mathcal{B}$. \EDmodified{We therefore have that $\{\sigma^{-1}\circcomp B\circcomp \sigma: B\in\mathcal{B}\} \supseteq \mathcal{B}$.  Finally, Lemma \ref{thm:app:subsetAPI} allows us to conclude that $\mathcal{B}$ is asset permutation invariant.}

\paragraph{Step 5 - WaveCorr is asset permutation invariant:} 
Combing Step 1 to 4 with Proposition \ref{thm:corrlayer}, we arrive at the conclusion that the architecture presented in Figure \ref{fig:generalArc} is asset permutation invariant since it is composed of a sequence of asset permutation invariant blocks.\qed

\subsection{Correlation Layer in \cite{zhang2020cost} Violates Asset Permutation Invariance}\label{sec:app:CSPPN}

\oldText{
\subsection{Old section}
Figure \ref{fig:corr} presents an example of the correlation layer in \cite{zhang2020cost}. Intuitively, the problem with this structure is that correlation information is only partially extracted by the model. More specifically, by moving the kernel of size 5 in this example, the convolution can only exploit the correlation present in contiguous subsets of the 5 assets. In particular, the first application only considers the first three assets, the second only considers the first four assets, and so on. On the contrary, in our proposed correlation layer, the correlation of an asset with all other assets available for investment is systematically  considered.

\subsection{New section}
}

Assuming for simplicity that $m$ is odd, the \quoteIt{correlational convolution layer} proposed in \cite{zhang2020cost} takes the form of the following set of functions:
\[\mathcal{B}:=\{B_{w,b}:\mathcal{W}\in\myRe^{m\times d\times d},b\in\myRe\}\]
where
$$B_{w,b}(\mathcal{T})[i,j,k]:= \sum_{\ell=1}^m \sum_{k'=1}^d {\cal T}[i-(m+1)/2+\ell,j,k'] \mathcal{W}[\ell,k,k']  + b, \;\; \begin{array}{l}\forall\,i=1,\dots,m\\\forall j=1,\dots,h\\\forall k=1,\dots,d\end{array}\,,$$
where $\mathcal{T}[i',:,:]:=0$ for all $i'\not\in \{1,\dots,m\}$ to represent a zero padding. Figure \ref{fig:corr} presents an example of this layer when $m=5$, $h=1$, and $d=1$. One can already observe in this figure that correlation information is only partially extracted for some of the assets, e.g. the convolution associated to asset one (cf. first row in the figure) disregards the influence of the fifth asset. While this could perhaps be addressed by using a larger kernel, a more important issue arises with this architecture, namely that the block does not satisfy asset permutation invariance.

\begin{figure}[h!]
	\centering
	\includegraphics[scale=.5]{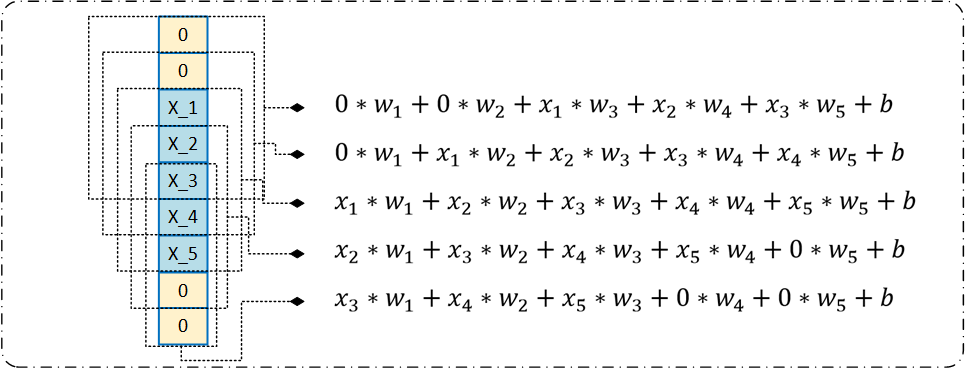}
	\caption{An example of the correlation layer in \cite{zhang2020cost}'s work over 5 assets}
	\label{fig:corr}
\end{figure}

\begin{proposition}
The correlational convolution layer block used in \cite{zhang2020cost} violates asset permutation invariance already when $m=5$, $h=1$, and $d=1$.
\end{proposition}

\begin{proof}
When $m=5$, $h=1$, and $d=1$, we first clarify that the correlational convolution layer from \cite{zhang2020cost} is associated with the following set of functions:
\[\mathcal{B}:=\{B_{w,b}:w\in\myRe^{5},b\in\myRe\}\]
where
$$B_{w,b}(\mathcal{T})[i]:= \left\{\begin{array}{ll}w_3\mathcal{T}[1]+w_4\mathcal{T}[2]+w_5\mathcal{T}[3]+b&\mbox{if $i=1$}\\
w_2\mathcal{T}[1]+w_3\mathcal{T}[2]+w_4\mathcal{T}[3]+w_5\mathcal{T}[4]+b&\mbox{if $i=2$}\\
w_1\mathcal{T}[1]+w_2\mathcal{T}[2]+w_3\mathcal{T}[3]+w_4\mathcal{T}[4]+w_5\mathcal{T}[5]+b&\mbox{if $i=3$}\\
w_1\mathcal{T}[2]+w_2\mathcal{T}[3]+w_3\mathcal{T}[4]+w_4\mathcal{T}[5]+b&\mbox{if $i=4$}\\
w_1\mathcal{T}[3]+w_2\mathcal{T}[4]+w_3\mathcal{T}[5]+b&\mbox{if $i=5$}
\end{array}\right.,$$
where we shortened the notation $\mathcal{T}[i,1,1]$ to $\mathcal{T}[i]$.
Let's consider the asset permutation operator that inverts the order of the first two assets: $\pi(1)=2$, $\pi(2)=1$, and $\pi(i)=i$ for all $i\geq 3$. We will prove our claim by contradiction. Assuming that $\mathcal{B}$ is asset permutation invariant, it must be that for any fixed values $\bar{w}$ such that $\bar{w}_4\neq \bar{w}_1$, there exists an associated pair of values $(w',b')$ that makes $B_{w',b'}\equiv\sigma^{-1}\circcomp B_{\bar{w},0}\circcomp\sigma$. In particular, the two functions should return the same values for the following three \quoteIt{tensors}: $\mathcal{T}_0[i]:=0$, $\mathcal{T}_1[i]:=\1\{i=1\}$, and at $\mathcal{T}_2[i]:=\1\{i=2\}$. The first implies that $b'=0$ since
\[b'=B_{w',b'}(\mathcal{T}_0)[1]=\sigma^{-1}(B_{\bar{w},0}(\sigma(\mathcal{T}_0)))[1]=0\,.\]
However, it also implies that:
\[w_2'=B_{w',0}(\mathcal{T}_1)[2]=\sigma^{-1}(B_{\bar{w},0}(\sigma(\mathcal{T}_1)))[2]=B_{\bar{w},0}(\mathcal{T}_2)[1]=\bar{w}_4\]
and that 
\[
w_2'=B_{w',0}(\mathcal{T}_2)[3]=\sigma^{-1}(B_{\bar{w},0}(\sigma(\mathcal{T}_2)))[3]=B_{\bar{w},0}(\mathcal{T}_1)[3]=\bar{w}_1\,.\]
We therefore have a contradiction since $\bar{w}_4=w_2'=\bar{w}_1\neq \bar{w}_4$ is impossible. We must therefore conclude that $\mathcal{B}$ was not asset permutation invariant.
\end{proof}

We close this section by noting that this important issue cannot simply be fixed by using a different type of padding, or a larger kernel in the convolution. Regarding the former, our demonstration made no use of how padding is done. For the latter, our proof would still hold given that the fixed parameterization $(\bar{w},0)$ that we used would still identify a member of the set of functions obtained with a larger kernel.


\subsection{Augmented policy network to accelerate training}\label{sec:app:input}

We detail in this section how the structure of the portfolio management problem \eqref{mainprob} can be exploited for a more efficient implementation of a policy network, both in terms of computation time and hardware memory. This applies not only to the implementation of WaveCorr policy network but also policy networks in \cite{jiang2017deep} and \cite{zhang2020cost}. In particular, given a a multiperiod objective as in \eqref{mainprob}, calculating the gradient  
$\nabla_{\theta} SR$ involves the step of generating a sequence of actions $a_0,a_1,...,a_{T-1}$ from a sample trajectory of states $s_0,s_1,...,s_{T-1} \in \mathbb{R}^{m \times h \times d}$ over a planning horizon $T$, where $m:$ the number of assets, $h:$ the size of a lookback window, $d:$ the number of features. 
The common way of implementing this is to create a tensor ${\cal T}$ of dimension $m \times h \times d \times T$ from $s_0,...,s_{T-1}$ and apply a policy network $\mu_{\theta}(s)$ to each state $s_t$ in the 
tensor ${\cal T}$ so as to generate each action $a_t$. Assuming for simplicity of exposition that the state is entirely exogenous, this procedure is demonstrated 
in Figure \ref{fig:fastVSslowConv}(a), where 
a standard causal convolution with $d=1$ and kernel size of 2 is applied. In this procedure, the memory used to store the tensor ${\cal T}$ and the computation time taken to generate all actions $a_0,...,a_{T-1}$ grow linearly in $T$, which become significant for large $T$. It is possible to apply the policy network $\mu_{\theta}(s)$ to generate all the actions $a_0,...,a_{T-1}$ more efficiently than the procedure described in Figure \ref{fig:fastVSslowConv}(a). Namely, in our implementation, we exploit the sequential and overlapping nature of sample states $s_0,...,s_{T-1}$ used to generate the actions $a_0,...,a_{T-1}$, which naturally arises in the consideration of a multiperiod objective. Recall firstly that each sample state $s_t \in \mathbb{R}^{m \times h \times d}$, $t \in \{0,...,T-1\}$, is obtained from a sample trajectory, denoted by ${\cal S}\in \mathbb{R}^{m\times (h+T-1) \times d}$, where
$s_t = {\cal S}[:,t+1:t+h,:]$, $t=0,...,T-1$. Thus, between any $s_t$ and $s_{t+1}$, the last $h-1$ columns in $s_t$ overlap with the first $h-1$ columns in $s_{t+1}$. The fact that there is a significant overlap between any two consecutive states $s_t,s_{t+1}$ hints already that processing each state $s_{t+1}$ separately from $s_t$, as shown in Figure \ref{fig:fastVSslowConv}(a), would invoke a large number of identical calculations in the network as those that were already done in processing $s_t$, which is wasteful and inefficient. To avoid such an issue, we take an augmented approach to apply the policy network. The idea is to use a sample trajectory ${\cal S}$ directly as input to an augmented policy network $\vec{\mu}_{\theta}:\mathbb{R}^{m\times (h+T-1) \times d}\rightarrow\mathbb{R}^{m\times T}$, which reduces to exactly the same architecture as 
the policy network $\mu_{\theta}(s_t)$ when generating only the $t$-th action. Figure \ref{fig:fastVSslowConv}(b) presents this augmented policy network $\vec{\mu}_{\theta}(S)$ for our example, and how it can be applied to a trajectory ${\cal S}$ to generate all actions $a_0,...,a_{T-1}$ at once. One can observe that the use of an augmented policy network allows the intermediate calculations done for each state $s_t$ (for generating an action $a_t$) to be reused by the calculations needed for the other states (and generating other actions). With the exact same architecture as the policy network $\mu_{\theta}(s)$, the augmented policy network $\vec{\mu}_{\theta}(S)$, which takes a trajectory with width $h+T-1$ (thus including $T$ many states), would by design generate $T$ output, each corresponds to an action $a_t$. This not only speeds up the generation of actions $a_0,...,a_{T-1}$ significantly but also requires far less memory to store the input data, i.e. the use of a tensor with dimension $(m \times (h+T) \times d)$ instead of $m\times h\times d\times T$. The only sacrifice that is made with this approach is regarding the type of features that can be integrated. For instance, we cannot include features that are normalized with respect to the most recent history (as done in \cite{jiang2017deep}) given that this breaks the data redundancy between two consecutive time period. Our numerical results however seemed to indicate that such restrictions did not come at a price in terms of performance.

\oldText{
In order to reduce the required hardware (GPU or CPU) memory for storing and processing the input data during training and to speed up the forward pass, which is to calculate the actions for a certain number of sequential days in the planning horizon $T$ and evaluate the SR objective function \eqref{mainprob}, we propose an approach that relies on the sequential order of actions the multiperiod objective. This approach can be employed in our waveCorr policy network as well as our benchmarks \cite{jiang2017deep,zhang2020cost}. In particular, notice that the state space for every two sequential actions are significantly overlapping with only one period difference. We can exploit this and use the calculations that are performed by the model for every action to come up with the next action while avoiding to repeat most part of these calculations. To elaborate more on this, considering $S \in \mathcal{S}$ to be the whole set of training samples, recall how the actions $a_1,...,a_{T-1}$ can be generated by using a $T$ dimensional mini-batch of states $s_1,...,s_T$, where $s_t = S[1:m,t-h:t,1:d]$, as usually done in practice disregarding the sequential pattern of actions and states. To achieve this, we need to assemble an input tensor of dimension $m \times h \times d \times T$, which represents the mini-batch of states and then feed the tensor into the network as shown in Figure \ref{fig:fastVSslowConv}(a). To compute the action at time $t$, $a_t$, the corresponding slice of the input tensor (i.e. a slice of dimension $m \times h \times d$ representing $s_t = S[1:m,t-h:t,1:d]$) should be processed by the network, where a large portion of these operations are already performed when computing $a_{t-1}$ using $s_{t-1} = S[1:m,t-h-1:t-1,1:d]$. We can avoid this redundancy by simply assembling the input tensor to be of a different dimensionality, as shown in Figure \ref{fig:fastVSslowConv}(b). This figure represents a standard causal convolution with $d=1$ and kernel size of $2$. In particular, instead of stacking many states $s_t$, we can consider a general state of dimension $m \times (h+T) \times d$ that will simultaneously provide us with $T$ sequential actions $a_1,...,a_{T-1}$, while making sure that all intermediate calculations are reused as much as possible and every action $a_t$ is computed by considering the corresponding state $s_t=S[1:m,t-h:t,1:d]$. The only difference between these two approaches is the dimension of the input tensors, while the networks are exactly similar. Apart from efficient use of processors in computing the actions, the significant reduction in the dimension of the input tensor improves memory usage during training. The only sacrifice that is made with this approach is regarding the type of features that can be integrated in this approach. For instance, we cannot include features that are normalized with respect to the most recent history (as done in \cite{jiang2017deep}) given that this breaks the data redundancy between two consecutive time period. Our numerical results however seemed to indicate that such restrictions did not come at a price in terms of performance.}


\oldText{
It is possible to reduce the memory usage of the model and speed up the learning process when employing policy networks from \cite{jiang2017deep,zhang2020cost}, or WaveCorr with a multiperiod objective as in \eqref{mainprob}. The classical approach used in training is to instantiate $T$ copies of the policy network that share the same parameters and stack them together to form a \quoteIt{multiple network mini-batch}. The SR objective function, together with a sampled trajectory, can then be used to derive a stochastic gradient step. This is however wasteful for two reasons: 1) it requires to assemble the input data for the whole stack of networks (i.e. $m\times h\times d\times T$); 2) inside each NN in the stack, a large number of computations (e.g. inside convolution layers) end up being exactly repeated. To reduce such wasteful storage and calculations, we perform training on an equivalent single augmented multi-period policy network producing all of the actions $(a_0,\,a_2,\,\dots,\,a_{T-1})$ simultaneously, while making sure that all intermediate calculations are reused as much as possible. Figure \ref{fig:fastVSslowConv} illustrates this procedure for a standard causal convolution network with $d=1$. Figure \ref{fig:fastVSslowConv}(a) presents the classical approach that construct the multiple network mini-batch.
 This approach is wasteful as the only difference between each of the $T$ layers is a single column of data, associated to the most recent time point.
 Figure \ref{fig:fastVSslowConv}(b) presents our suggested augmented structure, which exploits redundancy to compress down the total number of operations and storage needed to produce this $T$ period policy. Already from an input data management point of view, this reduction translates into using a tensor of $(m \times (h+T) \times d)$ instead of $m\times h\times d\times T$.
 The only sacrifice that is made with this approach is regarding the type of features that can be integrated in this approach. For instance, we cannot include features that are normalized with respect to the most recent history (as done in \cite{jiang2017deep}) given that this breaks the data redundancy between two consecutive time period. Our numerical results however seemed to indicate that such restrictions did not come at a price in terms of performance.}

\begin{figure}[h]
	\centering
	\begin{subfigure}{\mynewwidth\textwidth}
	
	\includegraphics[scale=.15]{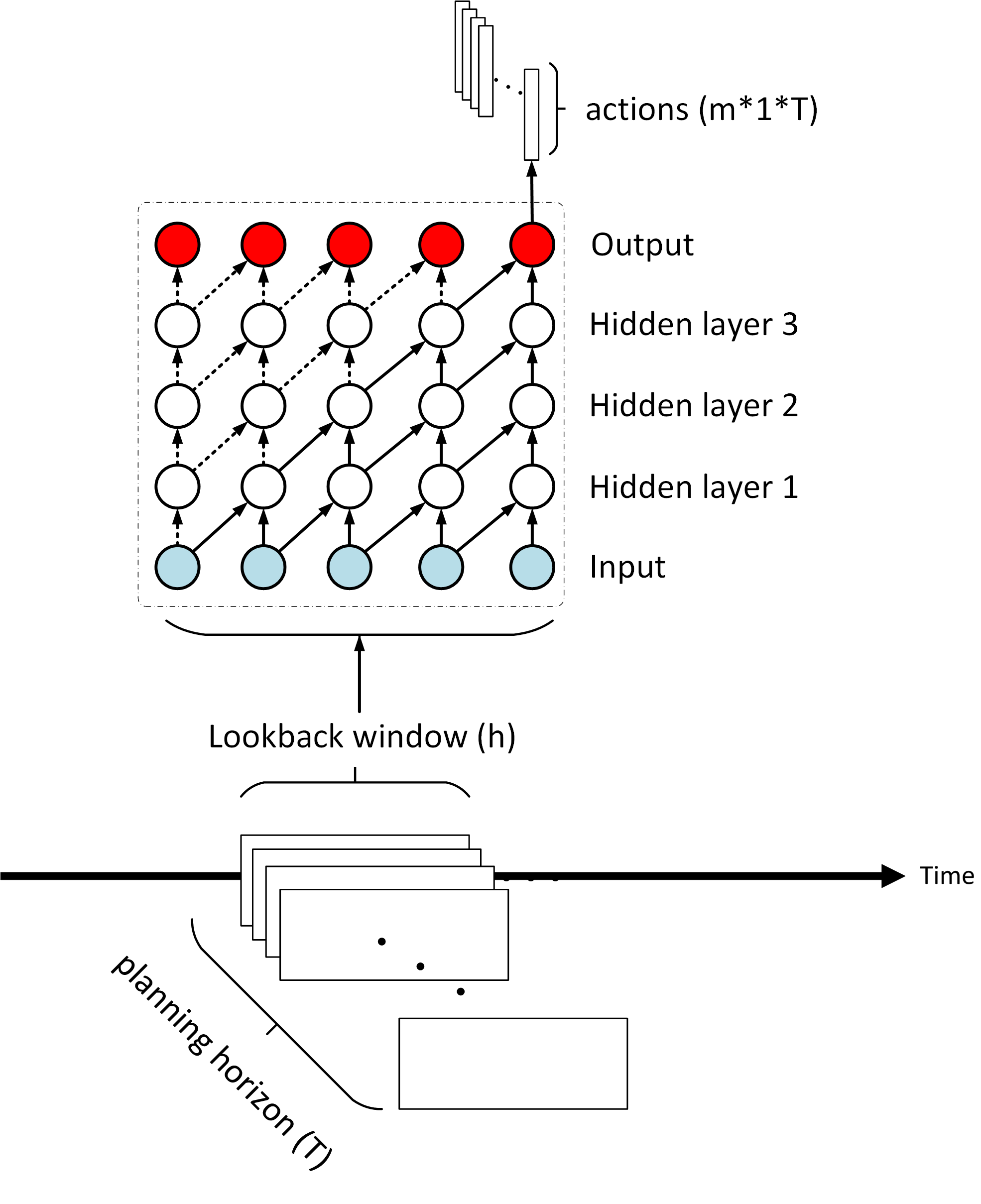}
		\caption{$\mu_{\theta}(s)$ applied to each state separately
		}
	\end{subfigure}
	\begin{subfigure}{\mynewwidth\textwidth}
		\includegraphics[scale=.15]{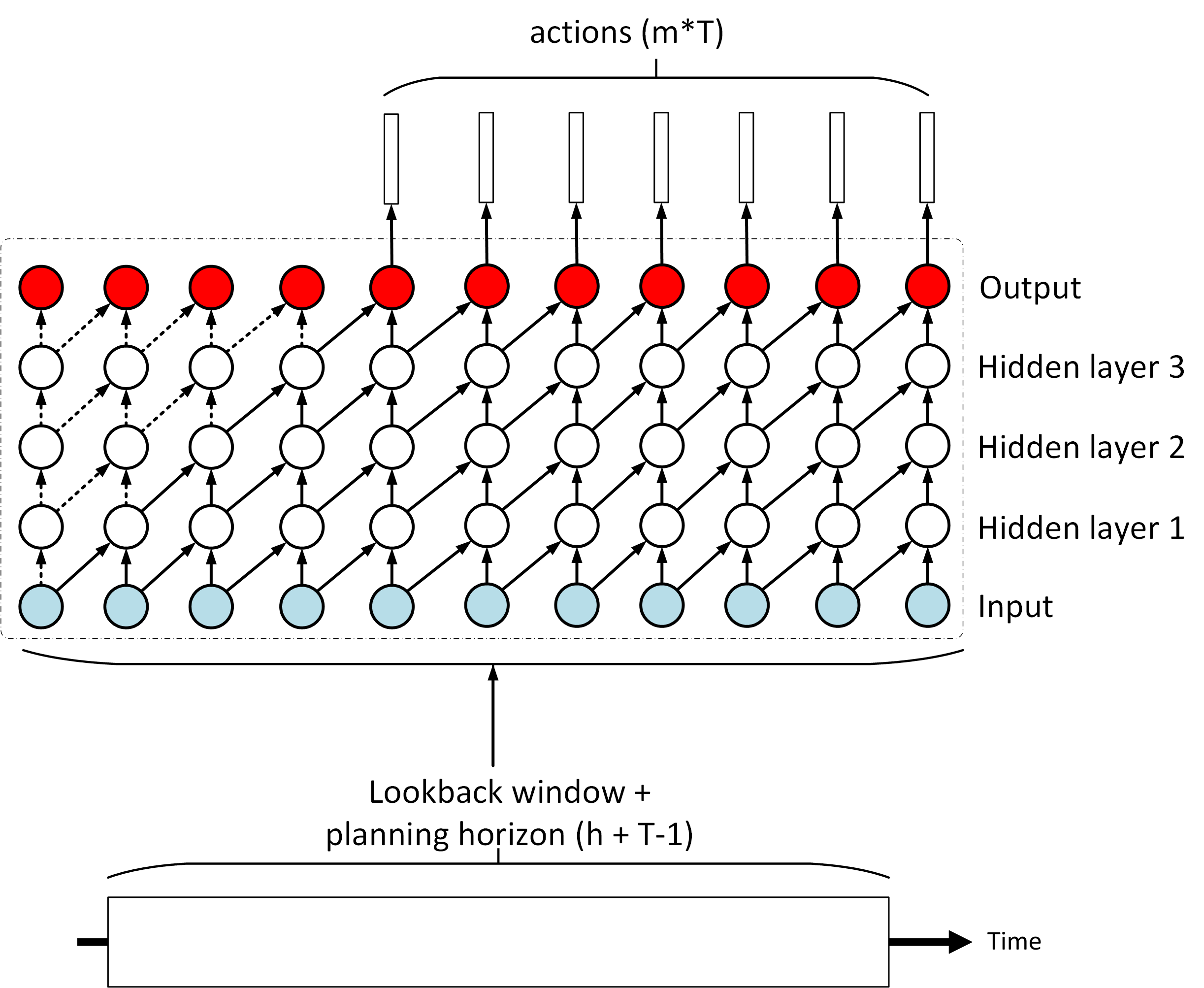}
		\vspace{1.1cm}
		\caption{$\vec{\mu}_{\theta}(S)$ applied to the full trajectory
		}
	\end{subfigure}
	
	\caption{Comparison between the use of policy network $\mu_{\theta}(s)$ and of the augmented policy network $\vec{\mu}_{\theta}(S)$ }
	\label{fig:fastVSslowConv}
\end{figure}

\newpage
\subsection{Hyper-parameters Selection}\label{sec:app:hyperParams}

\begin{table}[h]
\caption{List of Selected Hyper-parameters.}\label{table:hyperparams}
\label{tbl:hyperparams}
\centering
\begin{tabular}{lcccc}
\toprule
Hyper-parameter	&	Search range & WaveCorr & CS-PPN & EIIE	\\
\midrule
Learning rate	&	$\{5\times 10^{-5},\;10^{-4},\;  10^{-3},\;5\times10^{-3}\}$&$5\times 10^{-5}$&$5\times 10^{-5}$&$10^{-4}$	\\
Decay rate	&	$\{0.9999,\;0.99999,\;1\}$	&$0.99999$&$0.99999$&$1$\\
Minimum rate	&	$\{\; 10^{-6},\;10^{-5}\}$	&$10^{-5}$&$10^{-5}$&$10^{-5}$\\
Planning horizon $T$	&	$\{32,\; 64\}$	&$32$&$32$&$32$\\
Look back window size $h$	&	$\{32,\; 64\}$	&$32$&$32$&$32$\\
Number of epochs & $[0,\;\infty)$& 5000 & 5000 & 5000\\
\bottomrule
\end{tabular}
\end{table}

\newpage
\subsection{Additional results}\label{sec:app:addRes}

\subsubsection{Comparative Study}

\removed{
\begin{figure}[h!]
	\centering
	\begin{subfigure}{\mynewwidth\textwidth}
		\includegraphics[scale=\mynewscale]{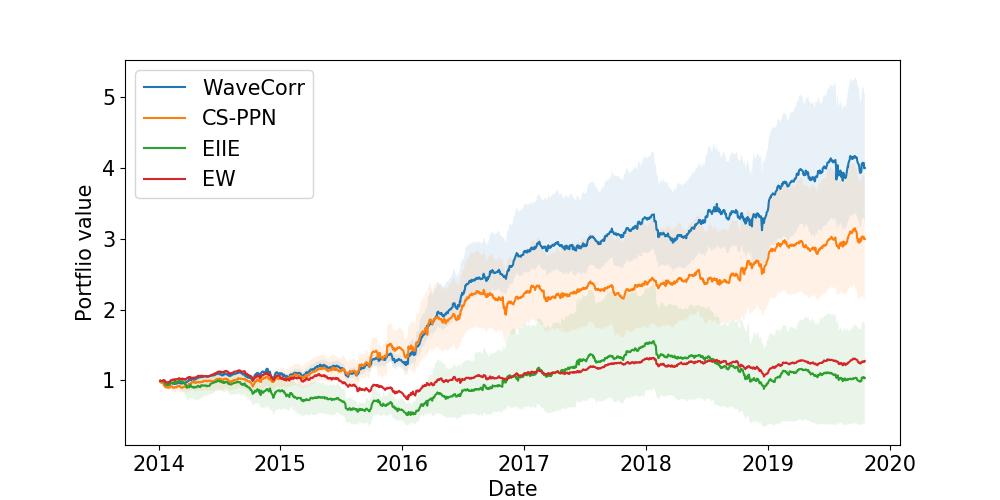}
		\caption{Can-data}
	\end{subfigure}
	\begin{subfigure}{\mynewwidth\textwidth}
		\includegraphics[scale=\mynewscale]{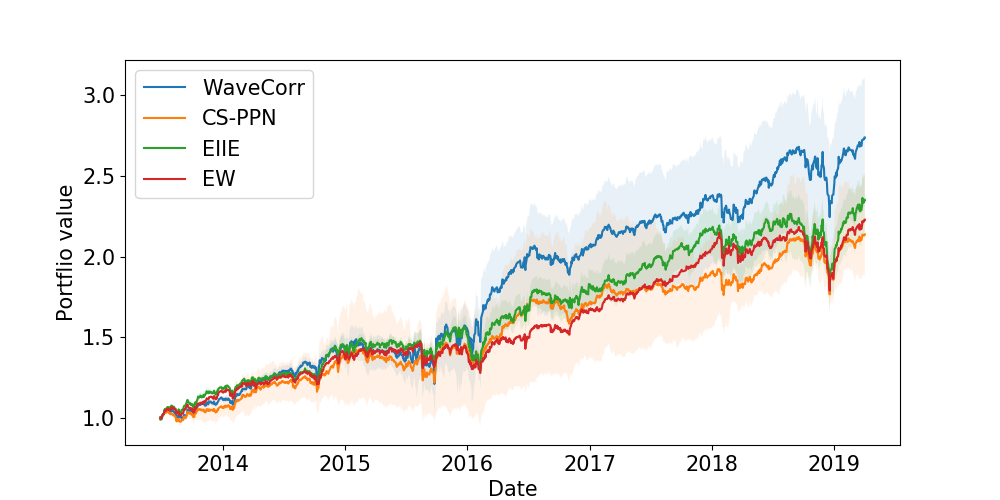}
		\caption{US-data}
	\end{subfigure}
	
 	\begin{subfigure}{\mynewwidth\textwidth}
 		\includegraphics[scale=\mynewscale]{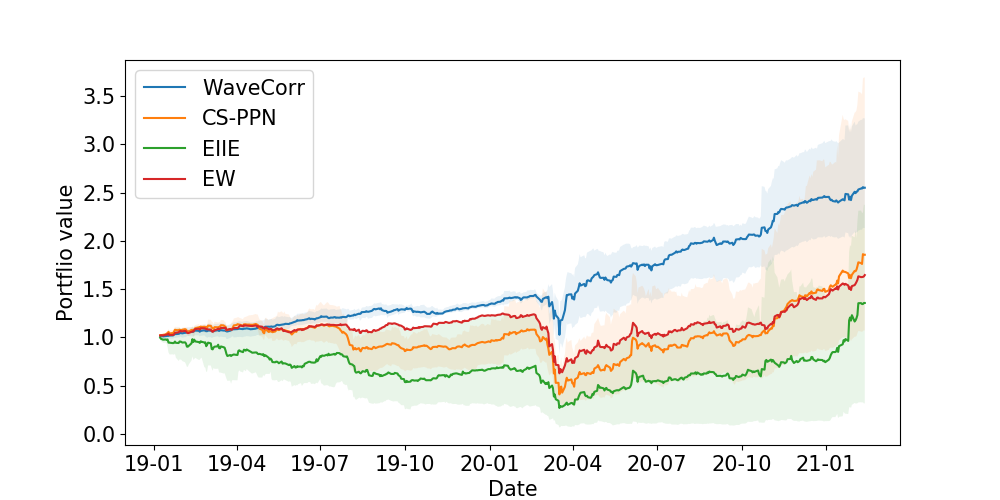}
 		\caption{Covid-data}
 	\end{subfigure}
	
	\caption{Comparison of the wealth accumulated (a.k.a. portfolio value) by WaveCorr, CS-PPN, EIIE, and EW over Can-data, US-data, and Covid-data.
	}
	\label{fig:mainexperiment}
	
\end{figure}}

\begin{figure}[h!]
	\centering
	\begin{subfigure}{\mybigwidth\textwidth}
		\includegraphics[scale=\mybigscale]{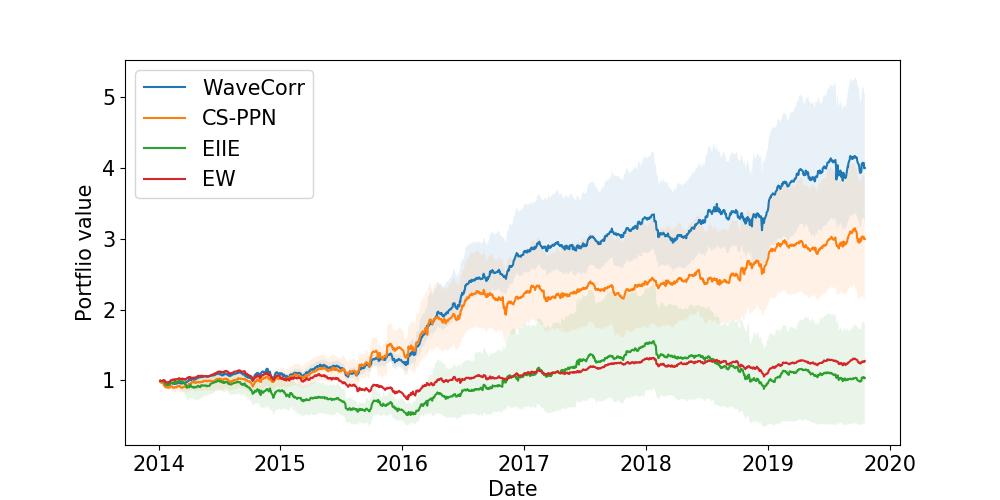}
		\caption{Can-data}
	\end{subfigure}
	
	\begin{subfigure}{\mybigwidth\textwidth}
		\includegraphics[scale=\mybigscale]{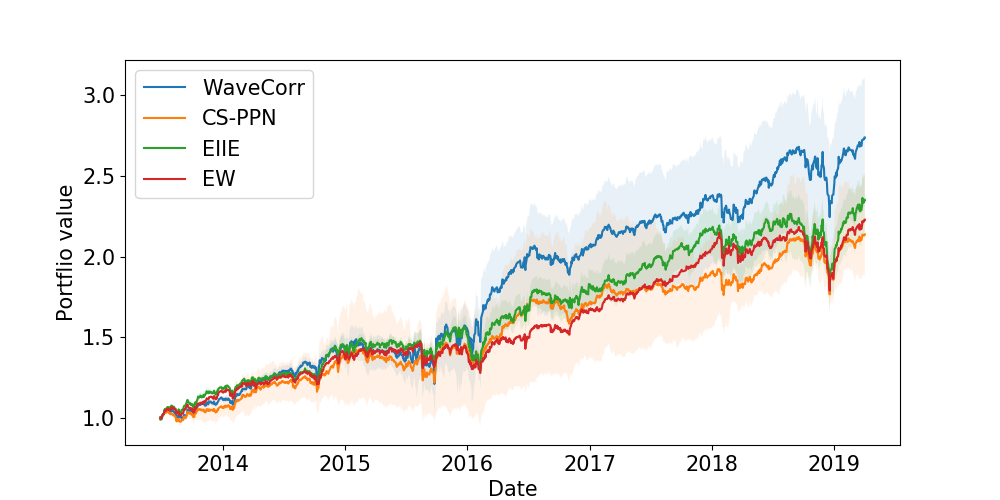}
		\caption{US-data}
	\end{subfigure}
	
 	\begin{subfigure}{\mybigwidth\textwidth}
 		\includegraphics[scale=\mybigscale]{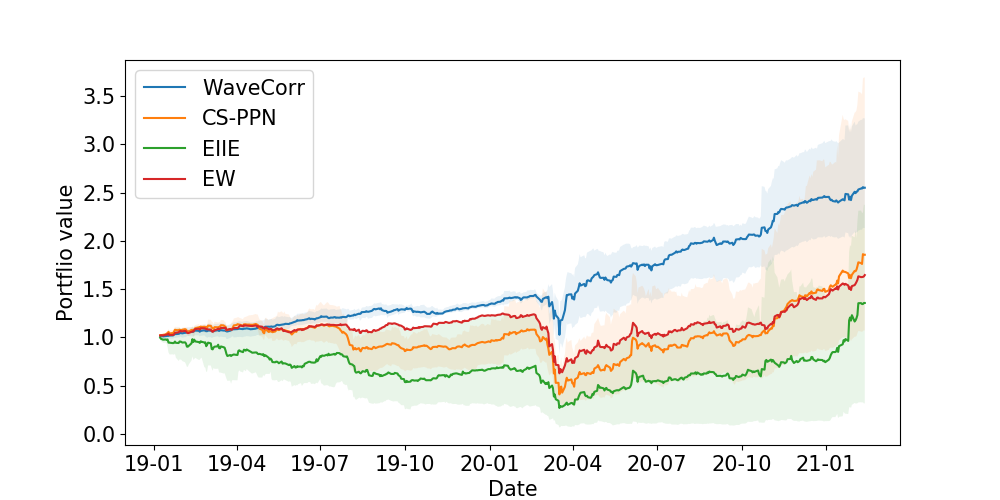}
 		\caption{Covid-data}
 	\end{subfigure}
	
	\caption{Average (solid curve) and range (shaded region) of out-of-sample wealth accumulated by WaveCorr, CS-PPN, EIIE, and EW over 10 experiments using Can-data, US-data, and Covid-data.
	}
	\label{fig:mainexperiment}
\end{figure}

\newpage
\subsubsection{Sensitivity to number of assets}

\begin{figure}[h!]
	\centering
	
\removed{
	\begin{subfigure}{\mynewwidth\textwidth}
		\includegraphics[scale=\mynewscale]{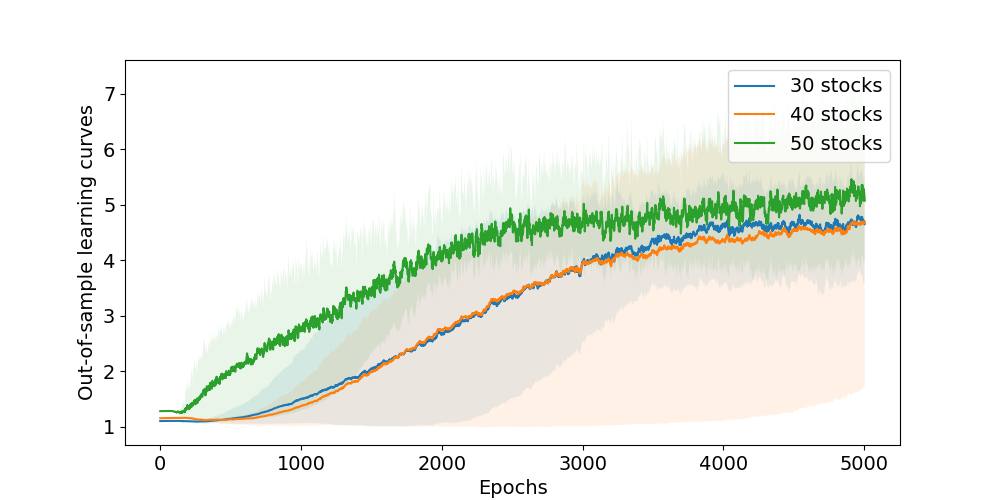}
		\caption{Out of sample learning curve, CS-PPN}
	\end{subfigure}
	\begin{subfigure}{\mynewwidth\textwidth}
		\includegraphics[scale=\mynewscale]{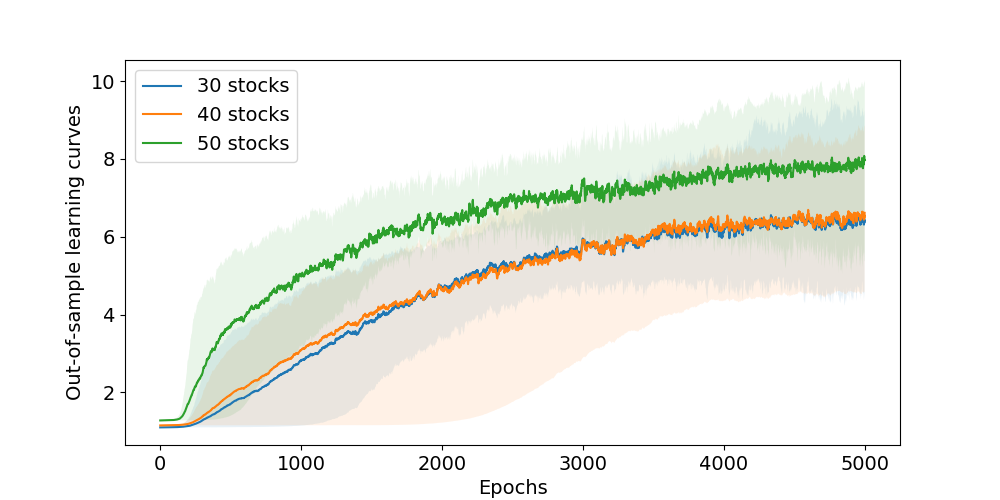}
		\caption{Out of sample learning curve, WaveCorr}
	\end{subfigure}}
	
	\begin{subfigure}{\mynewwidth\textwidth}
		\includegraphics[scale=\mynewscale]{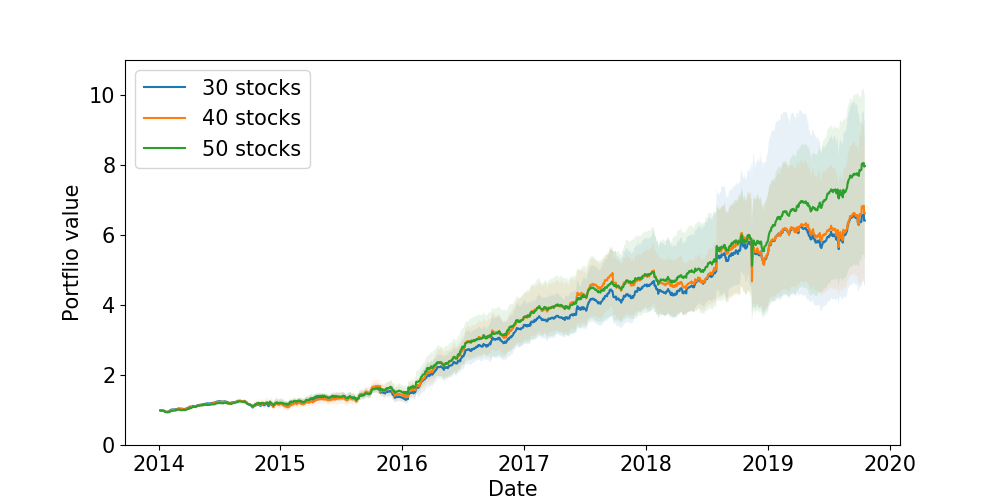}
		\caption{WaveCorr}
	\end{subfigure}
	\begin{subfigure}{\mynewwidth\textwidth}
		\includegraphics[scale=\mynewscale]{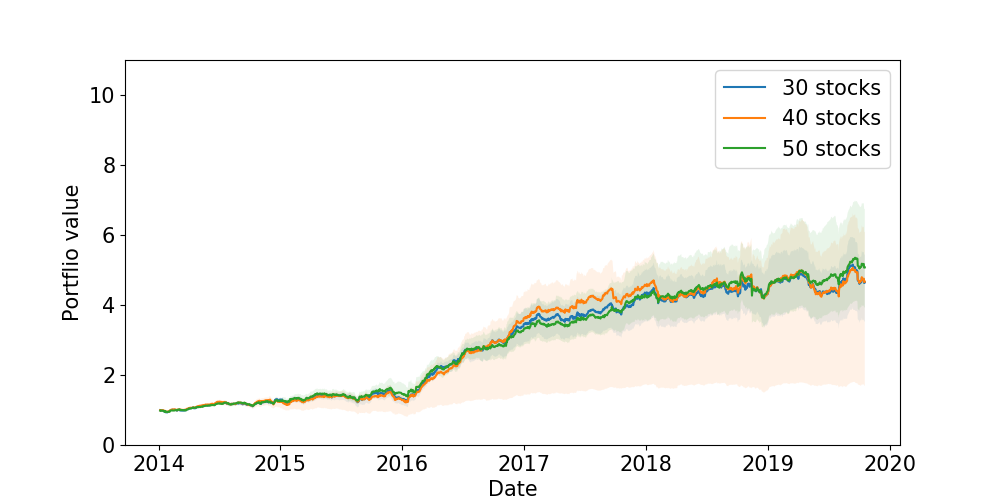}
		\caption{CS-PPN}
	\end{subfigure}

	\caption{Average (solid curve) and range (shaded region) of the out-of-sample wealth accumulated, on 10 experiments using Can-data, by WaveCorr and CS-PPN when increasing the number of assets.
	}
	\label{fig:wavecorr_compare_costsens_increase}
\end{figure}

\subsubsection{Performance comparison under maximum holding constraint}\label{sec:app:maxholding}

In practice, it is often required that the portfolio limits the amount of wealth invested in a single asset. This can be integrated to the risk-averse DRL formulation:
\begin{align*}
\bar{J}_F(\mu_\theta):=\mathbb{E}_{\substack{s_0\sim F\\s_{t+1}\sim P(\cdot|s_t,\mu_\theta(s_t))}}&[SR(r_0(s_0,\mu_\theta(s_0),s_1),...)]- \frac{M}{T} \sum_{t=0}^{T-1} \sum_{i=1}^{m}\max(0,\;\pw_t^i - w_{max})    
\end{align*}
where $w_{max}$ is the maximum weight allowed in any asset, and $M$ is a large constant. This new objective function penalizes any allocation that goes beyond $w_{max}$, which will encourage $\mu_\theta$ to respects the maximum weight allocation condition. The commission rates are considered to be $c_s=c_p=0.5\%$, and the experiments here are done over Can-data using the full set of 70 stocks, with a maximum holding of $20\%$. The results are summarized in Table \ref{tbl:real_environment} and illustrated in Figure \ref{fig:real_environment}. As noted before, we observe that WaveCorr outperforms CS-PPN with respect to all performance metrics.

\begin{figure}[h!]
	\centering
	\removed{\begin{subfigure}{\mynewwidth\textwidth}
		\includegraphics[scale=\mynewscale]{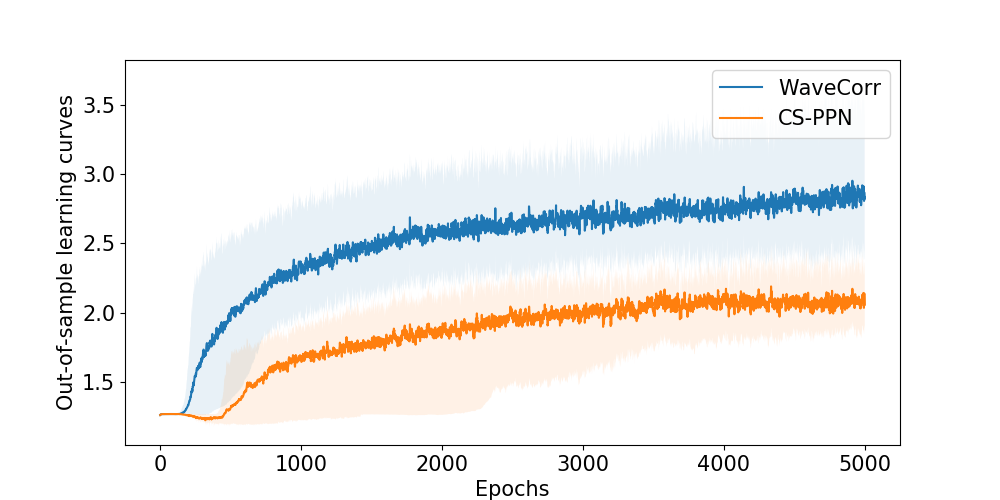}
		\caption{Out of sample learning curve}
	\end{subfigure}}
	\begin{subfigure}{\mynewwidth\textwidth}
		\includegraphics[scale=\mynewscale]{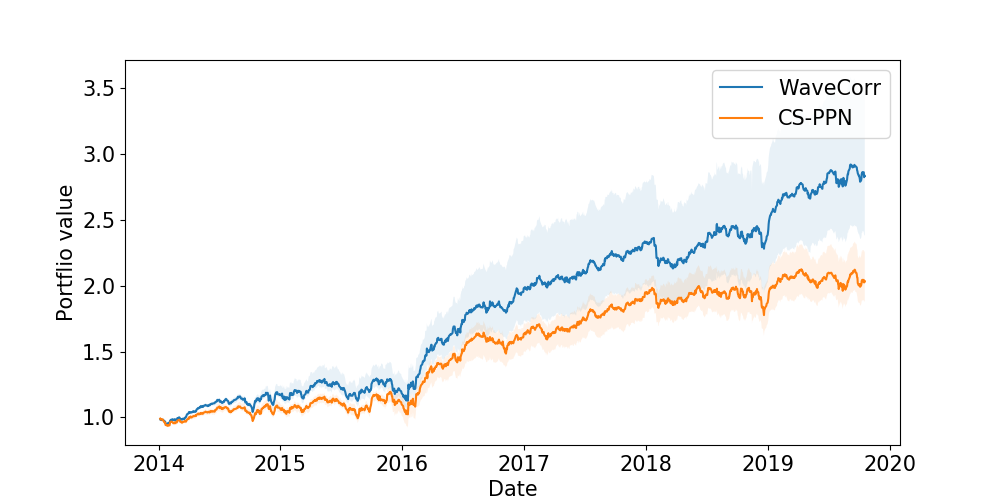}
		\caption{Out-of-sample cumulative returns}
	\end{subfigure}
	
	
	\caption{Average (solid curve) and range (shaded region) of the out-of-sample wealth accumulated, on 10 experiments using Can-data, by WaveCorr and CS-PPN under maximum holding constraint.
} 
	\label{fig:real_environment}
\end{figure}

\begin{table}[h]
\caption{The average (and standard dev.) performances when imposing a maximum holding constraints over 10 random initial NN weights in Can-data.}
\label{tbl:real_environment}
\centering
\begin{tabular}{l|cccccc}
\toprule
&	Annual return   	&	Annual vol      	&	SR     	&	MDD    	&	Daily hit rate  	&	Turnover        	\\
\midrule
WaveCorr	&	20\% (2\%)	&	13\% (0\%)	&	1.55 (0.18)	&	14\% (1\%)	&	53\% (1\%)	&	0.17 (0.01)	\\

CS-PPN	&	13\% (1\%)	&	13\% (1\%)	&	1.00 (0.15)	&	15\% (2\%)	&	50\% (1\%)	&	0.22 (0.03)	\\

\bottomrule

\end{tabular}
\end{table}


\bibliographystyle{unsrtnat}
\bibliography{References}

\end{document}